\def\confversion{0}
\def\ifconf{\ifnum\confversion=1}
\def\ifnotconf{\ifnum\confversion=0}

\documentclass[11 pt]{article}
\def\showauthornotes{1}
\def\showkeys{0}
\def\showdraftbox{1}

\usepackage{xspace,enumerate}
\usepackage{amsmath,amssymb}
\usepackage{amsthm}
\ifnum\confversion=0
	\usepackage[toc,page]{appendix}
\fi
\usepackage[T1]{fontenc}
\usepackage[utf8]{inputenc}

\usepackage{thmtools}
\usepackage{thm-restate}
\usepackage{color,graphicx}
\usepackage{boxedminipage}
\usepackage{makecell}
\usepackage{tabularx}
\usepackage{booktabs}
\usepackage{csquotes}

\ifnum\showkeys=1
\usepackage[color]{showkeys}
\fi

\definecolor{darkred}{rgb}{0.5,0,0}
\definecolor{darkgreen}{rgb}{0,0.35,0}
\definecolor{darkblue}{rgb}{0,0,0.55}

\usepackage[dvipsnames]{xcolor}

\usepackage[pdfstartview=FitH,pdfpagemode=UseNone,colorlinks,linkcolor=NavyBlue,filecolor=blue,citecolor=OliveGreen,urlcolor=NavyBlue,pagebackref]{hyperref}

\usepackage{enumitem}
\usepackage{comment}
\usepackage{todonotes}

\usepackage[capitalise,nameinlink]{cleveref}

\usepackage{microtype}

\usepackage{mathtools, bbm,dsfont}

\usepackage{palatino}
\usepackage[scaled=.95]{helvet}
\usepackage{eulerpx}

\DeclareFontEncoding{LGR}{}{}
\DeclareSymbolFont{sfitgreek}{LGR}{cmss}{m}{it}
\SetSymbolFont{sfitgreek}{bold}{LGR}{cmss}{bx}{it}
\DeclareMathSymbol{\sfpi}{\mathord}{sfitgreek}{`p}

\usepackage{tcolorbox}

\DeclareMathAlphabet{\mathpazocal}{OMS}{zplm}{m}{n}
\DeclareRobustCommand*{\mathcal}[1]{\mathpazocal{#1}}
\ifnum\showauthornotes=1
\newcommand{\Authornote}[3]{{\sf\small\color{#3}{[#1: #2]}}}
\newcommand{\Authorcomment}[2]{{\sf \small\color{gray}{[#1: #2]}}}
\newcommand{\Authorfnote}[2]{\footnote{\color{red}{#1: #2}}}
\else
\newcommand{\Authornote}[3]{}
\newcommand{\Authorcomment}[2]{}
\newcommand{\Authorfnote}[2]{}
\fi

\ifnum\showdraftbox=1

\else

\fi

\declaretheorem[numberwithin=section]{theorem}
\declaretheorem[sibling=theorem]{lemma}
\declaretheorem[sibling=theorem]{claim}
\declaretheorem[sibling=theorem]{proposition}
\declaretheorem[sibling=theorem]{fact}
\declaretheorem[sibling=theorem]{corollary}

\theoremstyle{definition}
\declaretheorem[sibling=theorem]{definition}
\declaretheorem[sibling=theorem]{remark}

\declaretheorem[sibling=theorem]{observation}
\declaretheorem[sibling=theorem]{example}

\newtheorem{algo}[theorem]{Algorithm}
\def\FullBox{\hbox{\vrule width 6pt height 6pt depth 0pt}}

\def\qed{\ifmmode\qquad\FullBox\else{\unskip\nobreak\hfil
\penalty50\hskip1em\null\nobreak\hfil\FullBox
\parfillskip=0pt\finalhyphendemerits=0\endgraf}\fi}

\def\qedsketch{\ifmmode\Box\else{\unskip\nobreak\hfil
\penalty50\hskip1em\null\nobreak\hfil$\Box$
\parfillskip=0pt\finalhyphendemerits=0\endgraf}\fi}

\def\matr#1{\mathsf{#1}}
\newcommand{\HS}{\mathrm{HS}}

\newcommand{\id}{\mathrm{I}}

\newcommand{\thresh}{\mathrm{Th}_{\scriptscriptstyle{A,t}}}
\newcommand{\threshold}[1]{\mathrm{Th}_{\scriptscriptstyle{A,#1}}}
\newcommand{\inv}[1]{#1\raisebox{1.15ex}{$\scriptscriptstyle-1$}}
\newcommand{\irrep}[1]{\widehat{G}}
\newcommand{\avgmatr}{\mathsf{N}}
\newcommand{\tr}{\mathrm{tr}}

\DeclareMathOperator{\Cay}{Cay}

\newcommand{\ones}{\mathds{1}}

\newcommand{\triv}{\mathrm{triv}}
\newcommand{\Irrep}{\mathrm{Irrep}}
\newcommand{\vrho}{{\vec{\rho}}}

\newcommand{\ep}{\varepsilon}
\newcommand{\rw}{{\sf RW}}

\newcommand{\sfM}{{\sf M}}

\newcommand{\sfC}{{\sf C}}

\newcommand{\inds}{\mathcal{S}}

\newcommand{\rwalk}[1]{{\textup{RW}}_{#1}}
\newcommand{\aprxone}{    \sum_{I\in \mathcal{I}_k} \lambda^{\sum_{j\in I}\Delta_j(\inds)}      }

\makeatletter
\newcommand{\stackalign}[1]{
	\vcenter{
		\Let@ \restore@math@cr \default@tag
		\baselineskip\fontdimen10 \scriptfont\tw@
		\advance\baselineskip\fontdimen12 \scriptfont\tw@
		\lineskip\thr@@\fontdimen8 \scriptfont\thr@@
		\lineskiplimit\lineskip
		\ialign{\hfil$\m@th\scriptstyle##$&$\m@th\scriptstyle{}##$\crcr
			#1\crcr
		}
	}
}

\let\latexcirc=\circ
\newcommand{\ccirc}{\mathbin{\mathchoice
  {\xcirc\scriptstyle}
  {\xcirc\scriptstyle}
  {\xcirc\scriptscriptstyle}
  {\xcirc\scriptscriptstyle}
}}
\newcommand{\xcirc}[1]{\vcenter{\hbox{$#1\latexcirc$}}}\let\circ\ccirc

\def\to{\rightarrow}

\def\epsilon{\varepsilon}

\def\phi{\varphi}

\def\implies{\Rightarrow}

\newcommand{\ie}{i.e.,\xspace}
\newcommand{\eg}{e.g.,\xspace}

\newcommand{\mper}{\,.}

\newcommand{\E}{{\mathbb E}}
\newcommand{\C}{{\mathbb C}}

\newcommand{\Z}{{\mathbb Z}}

\newcommand{\cE}{\mathcal{E}}

\newcommand{\indicator}[1]{\mathds{1}_{\{#1\}}}

\DeclarePairedDelimiter\parens{\lparen}{\rparen}
\DeclarePairedDelimiter\abs{\lvert}{\rvert}
\DeclarePairedDelimiter\norm{\lVert}{\rVert}
\DeclarePairedDelimiter\floor{\lfloor}{\rfloor}

\DeclarePairedDelimiter\braces{\lbrace}{\rbrace}
\DeclarePairedDelimiter\brackets{\lbrack}{\rbrack}
\DeclarePairedDelimiter\angles{\langle}{\rangle}
\DeclarePairedDelimiterXPP\lnorm[1]{}\lVert\rVert{_2}{#1}

\DeclareMathDelimiter{\given}
      {\mathbin}{symbols}{"6A}{largesymbols}{"0C}

\newcommand{\prob}{\mathsf{Pr}}
\newcommand{\Esymb}{\mathbb{E}}

\newcommand{\Psymb}{\mathrm{Pr}}

\DeclarePairedDelimiterXPP{\Prob}[1]
 {\prob}{\lparen}{\rparen}{}
 {\renewcommand{\given}{\;\delimsize\vert\nonscript\;\mathopen{}}#1}

\makeatletter
\def\Pr#1{%
    \ProbabilityRender{\Psymb}{#1}%
}

\def\Ex#1{%
    \ProbabilityRender{\Esymb}{#1}%
}

\def\condPE#1#2{%
	\@ifnextchar\bgroup
	{\ConditionalProbabilityRender{\widetilde{\Esymb}}{#1}{#2}}
	{\ProbabilityRender{\widetilde{\Esymb}}{#1 \given #2}}
}

\def\ConditionalProbabilityRender#1#2#3#4{
	\renderwithdist{#1}{#2}{#3 \given #4}	
}

\def\ProbabilityRender#1#2{%fancy probability command
  \@ifnextchar\bgroup%
  {\renderwithdist{#1}{#2}}
   {\singlervrender{#1}{#2}}
}
\def\singlervrender#1#2{%
   {\mathchoice
       {{#1}\brackets*{#2}}
       {{#1}[ #2 ]}
       {{#1}[ #2 ]}
       {{#1}[ #2 ]}
   }
}
\def\renderwithdist#1#2#3{%
   \@ifnextchar\bgroup
   {\superfancyrender{#1}{#2}{#3}}
   {\mathchoice
      {\underset{#2}{#1}\brackets*{#3}}
      {{#1}_{#2}[ #3 ]}
      {{#1}_{#2}[ #3 ]}
      {{#1}_{#2}[ #3 ]}
     }
}
\def\superfancyrender#1#2#3#4#5{
   \ensuremath{\mathchoice
      {\underset{#1}{{#1}}\left#4 #3 \right#5}
      {{#1}_{#2}#4 #3 #5}
      {{#1}_{#2}#4 #3 #5}
      {{#1}_{#2}#4 #3 #5}
   }
}
\makeatother

 \newcommand\SetSymbol[1][]{%
     \nonscript\:#1\vert
     \allowbreak
     \nonscript\:
     \mathopen{}}
  \DeclarePairedDelimiterX\Set[1]\{\}{%
     \renewcommand\given{\SetSymbol[\delimsize]}
     #1
}

\newcommand{\ip}[2]{\angles*{#1 , #2}}
 \newcommand{\set}[1]{\braces*{#1}}
\newcommand{\opnorm}[1]{\norm*{#1}_{\mathrm{op}}}

\newcommand{\op}{\mathrm{op}}

\usepackage{graphicx}
\usepackage{amsfonts}
\allowdisplaybreaks
\usepackage[top=1in, bottom=1in, left=1.25in, right=1.25in]{geometry}

\begin{document}

%\title{Fooling functions on non-abelian groups via expanders walks}

\title{Pseudorandomness of Expander Walks via \\
       Fourier Analysis on Groups}

%\title{Expander Walks via a Non-Commutative Fourier Analysis}

 \author{Fernando Granha Jeronimo\thanks{University of Illinois Urbana-Champaign, \;\tt{granha@illinois.edu}.}
                       \and
        Tushant Mittal\thanks{Stanford University, \; {\tt{tushant@stanford.edu}}. TM is a postdoctoral fellow supported by the NSF grants CCF-2143246 and CCF-2133154.} \and
        Sourya Roy\thanks{The University of Iowa, \; {\tt{sourya-roy@uiowa.edu}}. SR was partially supported by the Old Gold Summer Fellowship from The University of Iowa.}     }

\date{\today}

\date{}

\maketitle
%\vspace{-1.35cm}
%\draftbox
%\vspace{-0.5cm}
%% \thispagestyle{empty}

\thispagestyle{empty}
A long line of work has studied the pseudorandomness properties of walks on expander graphs. A central goal is to measure how closely the distribution over $n$-length walks on an expander approximates the uniform distribution of $n$-independent elements. One approach to do so is to label the vertices of an expander with elements from an alphabet $\Sigma$, and study closeness of the mean of functions over $\Sigma^n$, under these two distributions.  We say expander walks $\ep$-fool a function if the expander walk mean is $\ep$-close to the true mean. There has been a sequence of works studying this question for various functions, such as the XOR function, the AND function, etc. We show that:
\begin{itemize}
	\item The class of symmetric functions is $O(|\Sigma|\lambda)$-fooled by expander walks over any generic $\lambda$-expander, and any alphabet $\Sigma$ . This generalizes the result of Cohen, Peri, Ta-Shma [STOC'21] which analyzes it for $|\Sigma| =2$, and exponentially improves the previous bound of $O(|\Sigma|^{O(|\Sigma|)} \lambda)$, by Golowich and Vadhan~[CCC'22]. Moreover, if the expander is a Cayley graph over $\Z_{|\Sigma|}$, we get a further improved bound of $O(\sqrt{|\Sigma|}\lambda)$.
\end{itemize} 

\noindent Morever, when $\Sigma$ is a finite group $G$, we show the following for functions over $G^n$:

\begin{itemize}
	\item The class of symmetric class functions is $O\parens[\Big]{\frac{\sqrt{|G|}}{D}\lambda}$-fooled by expander walks over \enquote{structured} $\lambda$-expanders, if $G$ is $D$-quasirandom. 
%	\item In particular, for Abelian groups, it sharpens the above result to $O\parens[\big]{\sqrt{|G|}\cdot \lambda}$ for such structured graphs. \todo{Abelian confusing}
	\item We show a lower bound of $\Omega(\lambda)$ for symmetric functions for any finite group $G$ (even for \enquote{structured} $\lambda$-expanders). 
	\item We study the Fourier spectrum of a class of non-symmetric functions arising from \emph{word maps}, and show that they are exponentially fooled by expander walks.
\end{itemize}

 Our proof employs Fourier analysis over general groups, which contrasts with earlier works that have studied either the case of $\Z_2$ or $\Z$. This enables us to get quantitatively better bounds even for unstructured sets.

\newpage

%% \ifnotconf
%% \pagenumbering{roman}
\tableofcontents
\thispagestyle{empty}
%\newpage\todototoc\listoftodos
 \clearpage
%% \fi
%\endgroup
 
 %\newpage

\pagenumbering{arabic}
\setcounter{page}{1}

\section{Introduction}

Expander graphs are fundamental pseudorandom objects with a vast range of applications in computer
science and mathematics \cite{HooryLW06,Vadhan12,RL10}. These graphs combine two opposing properties of being
well-connected yet sparse. In many ways, they exhibit behavior that is surprisingly close to truly random,
thereby being used to replace randomness and yield several explicit constructions. For instance, explicit codes
 approaching the random guarantees of the Gilbert--Varshamov \cite{G52,V57} bound \cite{Ta-Shma17} or
the generalized Singleton \cite{ST20,Rot24} bound can be constructed using expanders \cite{JMST25}. Moreover, expanders
can be used to construct a variety of pseudorandom generators \cite{INW94,HH24}. 

Walks on expander graphs not only mix fast, but they are an important derandomization tool for the
Chernoff bound \cite{G93,GLSS18}, the hitting set property \cite{AKS87,TZ24}, etc. These tasks can be phrased
in a more general and unified way as how well expander walks fool a Boolean function $f \colon \set{0,1}^n \to \set{0,1}$.
In this setting, the vertices, $V_X$, of an expander graph $X$, are labelled with bits $\set{0,1}$ and instead of evaluating 
$f$ under the uniform distribution on $n$ bits, we evaluate it under the distribution on $\set{0,1}^n$
induced by taking a random walk on $X$ of length $n$. To quantify the error incurred in this process of replacing true randomness by expander walks, it is convenient to define $\mathcal{E}_X(f)$ as:
\begin{align*}
  \mathcal{E}_X(f) ~=~ \left| \E_{s \sim  V_X^n}\left[ f(s) \right] - \E_{s \sim \textup{RW}_n}\left[ f(s) \right] \right|\mper
\end{align*}
In this language, by choosing $f$ to be the AND function on $n$ bits, we recover the expander hitting set property
application. The choice of $f$ as a (suitable) threshold function on $n$ bits leads to the expander Chernoff bound. The choice of $f$ as the XOR function on $n$ bits (and a carefully constructed $X$) leads to the breakthrough
code construction of Ta-Shma~\cite{Ta-Shma17}.

Using this unified perspective, Cohen, Peri, and Ta-Shma \cite{CPT21} developed a systematic framework to analyze expander
random walks and obtain bounds on $\mathcal{E}_X(f)$. Their framework is based on Fourier analysis and exploits the fact
that any Boolean function $f \colon \set{0,1}^n \to \set{0,1}$ can be expressed in the Fourier basis as a linear combination
of characters, which are XOR functions in this case. They obtained bounds on $\mathcal{E}_X(f)$ in terms of the spectral
expansion $\lambda$ of the (normalized) adjacency matrix of $X$.

A series of works \cite{GV22,CMPPT22,G23} have since extended the \cite{CPT21} framework to functions of the form $f \colon \Sigma^n \to \mathbb{C}$, where $\Sigma$ is a finite alphabet. These works study \textit{symmetric functions}---functions that are invariant under any permutation of the input coordinates---and functions computed by restricted circuit classes such as $\mathsf{AC}^0$. Golowich and Vadhan~\cite{GV22} get the following bound for symmetric functions:  
\[
\abs{\cE_X(f)} ~\leq~ \parens{|\Sigma|^{O(|\Sigma|)}\cdot \lambda }.
\]
They asked whether such an exponential dependence on $|\Sigma|$ is optimal. In this work, we improve this bound to $O\parens{|\Sigma|\cdot \lambda }$ by viewing a function $f:\Sigma^n\to \C$ as a function $f:\Z_{\abs{\Sigma}}^n\to \C$ which enables us to use a Fourier basis. More interestingly, this change of perspective motivates us to consider graphs that can utilize this algebraic structure, such as Cayley graphs over $\Z_{\abs{\Sigma}}$. This idea helps us get a bound of $O\parens{ \sqrt{|\Sigma|}\cdot \lambda }$ for Cayley expanders, further improving upon our bound of $O\parens{|\Sigma|\cdot \lambda }$ for arbitrary expanders.

% Note that all prior results (and our improvement) only use the spectral expansion of $X$, and thus, give a bound for worst-case expanders. 

\paragraph{Functions over Groups}

More broadly, the above suggests investigating functions over general finite groups (instead of just $\Z_{\abs{\Sigma}}$) and considering expanders with a compatible algebraic structure. This opens up interesting directions to study:

%More importantly, this idea of endowing the set $\Sigma$ (and the expander $X$) with an algebraic structure 
\begin{enumerate}
	\item What novel classes of functions $f: G^n\to\C$  can we study that utilize the group operation, or have richer symmetries coming from the group? For instance, \textit{class functions}, \ie  functions that satisfy $f(x) = f(gxg^{-1})$ for every $x,g \in G$.   
%	\item For example, functions $f$ such that $f()$
%	\item How do we define and analyze generalizations of computational models such as $\mathsf{AC}^0$.
\item Can one obtain stronger bounds on $\cE_X(f)$ for such function classes when the expander $X$ has algebraic structure?
\item How can one utilize the pseudorandomness properties of the group itself, such as \textit{quasirandomness}?
\end{enumerate} 
While these questions are very natural in their own right, studying functions over general groups has been fruitful in the context of complexity theory. For instance, the famed Barrington's theorem \cite{Barrington89} effectively reinterprets a Boolean function as a function over the permutation group. More recently,~\cite{JMRW22} showed that one can obtain improved expanders by studying pseudorandomness for functions over the permutation group. 

\subsection{Our Results}
We initiate a study of this general setup and make progress in answering these questions. We (i) give a general lemma about the pseudorandomness for matrix-valued functions (over arbitrary expanders), which is needed to work with Fourier decompositions of non-abelian groups, (ii) analyze specific function classes of symmetric functions and \textit{word functions} over a group $G$, and (iii) we study symmetric \textit{class functions} over structured expanders, and show that the quasirandomness of the group $G$ synergizes with the randomness of the expander $X$ to yield improved bounds, and finally, (iv) we prove a lower bound for the fooling of symmetric functions even over such structured expanders.

%In the following section, we describe our results more formally, we first begin by considering arbitrary functions $f: G^n\rightarrow \C$ devoid of any additional structure. 

% generic expander case

\subsubsection{Pseudorandomness of generic expanders}
We begin by considering an abstract problem about expander walks. Let $X$ be an expander and consider a set of $k$ matrix-valued mean-zero functions: 
\[
\braces[\big]{f_j: X\rightarrow  \C^{d_j\times d_j} \mid j \in [k]},\; \quad \max_x\norm{f_j(x)}_{\op}\leq 1.
\] 
%to be any set of  functions on an expander $X$ such that  $$, and $\max_x\norm{f_j(x)}$.
Given an ordered subset $\inds = \set{i_1 < i_2<\cdots < i_{k-1}< i_k}$ of indices, we wish to study the expression,
\[
\cE_{X,\inds}\,(f_1\otimes\cdots \otimes f_k) ~:=~ \norm[\big]{\E_{\vec{x}\sim \rw_n} [f_1(x_{i_1})\otimes \cdots  \otimes f_k(x_{i_k}) ]}_\op
\]
Note that the above expression is identically $0$ if $X$ is a complete graph (with self-loops) as the functions, $f_i$, have zero mean. The goal is to show that above quantity is small when we have  a $\lambda$-expander. Analyzing special cases of such quantity is at the heart of several past works~\cite{CPT21,JMRW22,RR2024} that studied pseudorandomness of expander walks. 
%The functions 
%It can be shown that bounding $\mathcal{E}_X(f)$
%
%As we are working with a function $f: G^n\rightarrow \C$, using fourier analysis over  it is possible to decompose $f$ as linear combination of certain product functions. 
%Ideally, we would like to say the quantity $\mathcal{E}_X(f)$ is negligible for any function $f$. 
%We compare two distributions - (1) $\unif_n$  be the uniform distribution of $n$-independent vertices from an expander $X$ and $\rw_n$ be the uniform distribution on $n$-length random walks. Now, given a test function $f: G^n\rightarrow \C$, 
%\begin{center}
%we want to study $\Ex{f\circ \phi}$ behaves under the two distributions - $\unif_n$ and $\rw_n$ where, $f\circ \phi(\vec x):= f\parens{\phi(x_1),\dots,\dots, \phi(x_n) }$.
%\end{center}
%To minimize notation clutter, we supress the labeling map $\phi$ and keep it implict from this point. 
%We want to compare the mean of the function $f\circ \phi: G^n $
%$\rw_n$. 
%The setup of the result is as follows. We have a collection of $k$ matrix-valued functions 
%Let 
For the setting of matrix-valued functions, the result of~\cite{JMRW22} gives the following bound on the expression: 

\begin{theorem}[{~\cite{JMRW22}}]\label{theo:jmrw_intro}
	Let	$X$ be any $\lambda$-expander and $\{f_1,\dots, f_k\}$ be a collection of mean-zero normalized matrix-valued functions for $k \geq 2$. Then for any $k$-sized subset of indices, $\inds$,
	\[
	\cE_{X,\inds}\,(f_1\otimes\cdots \otimes f_k) ~\leq~ (2\lambda)^{\floor{\frac{k}{2}}}. 
	\]
\end{theorem}

The above result is agnostic to the set $\inds$ and gives a general worst-case bound. But this is too pessimistic when $\inds$ has large gaps. For instance, if $\inds = \{1,4,9\}$, the above result gives a bound of $\lambda$ which is far from the optimal (could be as small as $\lambda^8$). We prove a result that takes the structure of $\inds$ into account and gives an improved guarantee when $\inds$ has large gaps. 

%The key improvement in our result is that the pseudorandomness
%guarantee is improved when $\inds$ has large gaps. For instance if $\inds = \{1,4,9\}$, we get a bound of $\lambda^{3}$ instead of $\lambda$ which
%corresponds to the case when $\inds = \{1,2,3\}$.  
%%o the improvement in~\cite{CPT21}, for 
%%certain class of boolean valued functions, over the result of Ta-Shma~\cite{Ta-Shma17} 
%We have the following theorem. 
%For certain class of very specific boolean valued functions, 
%~\cite{CPT21} gave a more tight bound that takes  $\inds$
%into account. 

\begin{theorem}[Matrix-Valued Fooling, \cref{theo:tensor_fool}]\label{theo:main_tensor}
	Let $X$ be any $\lambda$-expander and $\{f_1,\dots, f_k\}$ be a collection of mean-zero  matrix-valued functions for $k \geq 2$. Then for any $k$-sized subset of indices, $\inds$,
	\[ 
	\cE_{X,\inds}\,(f_1\otimes\cdots \otimes f_k) 	~\leq~\lambda^{\eta(\inds)} ~\leq~  (4\lambda)^{\floor{\frac{k}{2}}}\mper  
	\]
	where $\eta(\inds)$ is an explicit function. 
\end{theorem}

The above bound is an improvement (over~\cref{theo:jmrw_intro}) when $\inds$ has large gaps.~\cref{theo:main_tensor} generalizes the result of~\cite{CPT21}, who achieved the same improvement over the result of Ta-Shma~\cite{Ta-Shma17} for $\{\pm 1\}$-valued functions. This quantitative improvement was crucially used by~\cite{CPT21} to prove their result about fooling functions over $\{0,1\}^n$, and we use it similarly to prove it for functions over an arbitrary alphabet, $\Sigma^n$.

The above theorem connects with fooling functions via Fourier analysis. Let $G$ be a group, and $f:G^n\to \C$ be any function. Consider a labeling\footnote{We only work with unbiased labelings, \ie those that induce the uniform distribution on $G$.} map, $\phi: X\rightarrow G$. Then, $f\circ \phi :X^n\to \C$, and the term we wish to bound is: 
\[
\mathcal{E}_X(f) ~:=~ \abs[\Big]{\Ex{\vec x \sim  \rw_n}{f \parens[\big]{\phi(x_1),\dots, \phi(x_n)}} - \Ex{\vec x \sim  \mathrm{Unif}_n}{f \parens[\big]{\phi(x_1),\dots, \phi(x_n)}}  }\mper
\]
Furthermore, since $f$ is a function on a product group, $G^n$, we can apply the general Fourier transform to express $f$ as a linear combination of matrix-valued tensor functions\footnote{\cref{theo:main_tensor} enables the possibility of working with orthonormal bases other than the Fourier basis. Any reasonably \enquote{flat} orthonormal basis where the basis elements satisfy certain pointwise bound and contains the invariant vector (i.e., the all $1$ vector) can be used.}. These tensor functions (when composed with $\varphi$) can be analyzed using~\cref{theo:main_tensor}. This can be seen as a generalization of the Fourier analytic approach of~\cite{CPT21}, who study symmetric functions over $\Z_2^n$. As a first application, we use our generalization from $\Z_2^n$ to $\Z_{|\Sigma|}^n$, to prove~\cref{theo:intro_main_fool}. Additionally, the ability to work with a general Fourier basis is utilized for other results where the function uses group structure for a given (potentially non-Abelian) group $G$, such as for \textit{group products}~(\cref{theo:word_func_intro}), and our general lower bound~(\cref{theo:lower_bound_intro}).

\subsubsection{Fooling symmetric functions and word functions}

We analyze the fooling of symmetric functions, ie functions $f:\Sigma^n\to \C$ that is invariant under permuting the input coordinates, for any finite alphabet $\Sigma$. 
% Finally, as we work with functions on finite groups,  if no additional structure is assumed on such functions, then effectively we are working with functions over finite sets.  For example set theoretic results for arbitrary set, $\Sigma$, can be recovered simply by considering the group $G=\Z_{|\Sigma|}$. Thus, by handling arbitrary finite group, one is also able to recover the same conclusion for finite sets. 
% for generic functions, sets are equiv to groups. Note that although the set theoretic setting may appear more general than the group setup, the set theoretic results for arbitrary set, $\Sigma$, can be recovered simply by considering the group $G=\Z_{|\Sigma|}$, provided that one can prove such result for arbitrary cyclic groups. In this paper, we take the group theoretic route, proving the following theorem that handles any finite group.

\begin{theorem}[Fooling symmetric functions, \cref{theo:main_fool}]\label{theo:intro_main_fool}
	Let $f$ be any symmetric function, $f:\Sigma^n \to \C$ where $\Sigma$ is any finite set. Let $X$ be a $\lambda$-expander such that $\lambda < \frac{1}{16e|\Sigma|}$. Then for any unbiased labelling of $X$ with $\Sigma$,
	\[  \abs{\cE_X(f)} ~\leq~ \parens{32e\lambda |\Sigma|}\cdot\norm{f}_2\quad.\]
	Moreover, if $\norm{f}_2 = o_n(1)$---for example, the weight indicator function which satisfies $\norm{f}_2 =n^{-1/4}$---one obtains a vanishing decay.
\end{theorem}
\noindent

This improves the previous best bound of $\parens{|\Sigma|^{O(|\Sigma|)}\cdot \lambda }$ due to
Golowich and Vadhan~\cite{GV22}. 
Our analysis relies on using a Fourier basis for such functions that can be obtained by viewing $\Sigma$ instead as $\Z_\Sigma$, and then applying~\cref{theo:main_tensor}. However, our proof is agnostic to this specific choice of group and can instead work with a Fourier basis over any group (of size $|\Sigma|$)  by using~\cref{theo:main_tensor}.

%Hence, we give a proof that works with a general Fourier basis, which is c. 

% consider symmetric functions $f$ on $\Sigma^n$ over finite alphabet $\Sigma$ and get the following bound:  
%\[
%\abs{\cE_X(f)} ~\leq~ \parens{|\Sigma|^{O(|\Sigma|)}\cdot \lambda }.
%\]

%\begin{corollary}
%	functions]\label{theo:intro_main_fool}
%	Let $f$ be any symmetric function, $f:\Sigma^n \to \C$ where $\Sigma$ is any finite group. Let $X$ be a $\lambda$-expander such that $\lambda < \frac{1}{16e|\Sigma|}$. Then for any unbiased labelling of $X$ with $\Sigma$,
%	\[  \abs{\cE_X(f)} ~\leq~ \parens{32e\lambda |G|}\cdot\norm{f}_2\quad.\]
%\end{corollary}

 \paragraph{Word Functions and Group Products} Going beyond symmetric functions, we analyze ``non-commutative'' monomial functions, which we call \textit{word} functions.
 
 \begin{definition}[Monomial word function]
	
	For an ordered subset $\inds \subseteq [n]$, a \textit{monomial word} is a map, $w_{\inds}: G^n\to G$, defined as $w_{\inds} = \prod_{s\in S} g_s^{e_s}$ where $e_S \in \{\pm 1\}$. A function $f:G^n \to \C$ is a \textit{monomial word function} of degree $k$, if $f = h(w_\inds(g_1,\cdots, g_n))$ for some $\inds$ of size $k$ and a function $h: G\rightarrow \C$.  
	%The \textit{degree} of a word is the size $|\inds|$.
%	 A word is \textit{monomial} if the variables are non-repeating and the exponent is $\pm 1$.  for a monomial word $w$ of degree $k$ and a function $h: G\rightarrow \C$. 
\end{definition}

%In a nutshell, this class of functions contains functions of the form $f(x_1,\ldots,x_n) = h(x_{i_1} \cdots x_{i_k})$, where $h \colon G \to \C$ is arbitrary.

%\textit{Word functions} are functions of the form $f(x_1,\cdots, x_n) = h(x_1x_2\cdots x_n)$ where $h: G\to \C$. These functions 

%In this section, we define a very natural class of functions, that we term \emph{word functions}. 

We give a complete characterization of the Fourier spectrum of \emph{monomial word functions}, and show that these have Fourier support on the highest level and thus are analogs of the PARITY function over $\Z_2^n$. Moreover, this support is sparse (see~\cref{lem:repsame}), and this enables us to use~\cref{theo:main_tensor}. We thus deduce that such functions are exponentially fooled by expander walks.

%
%The key result we show here is that the group product functions (more generally, word functions), as above, have their Fourier support entirely on the maximum possible degree irreps. Moreover, this support is also very sparse . Thus, we effectively prove a tight characterization of its Fourier spectrum that enables us to use~ \cref{theo:main_tensor}.  

\begin{theorem}[Fooling word functions, \cref{theo:word_func1}]\label{theo:word_func_intro}
	Let  $f: G^n\rightarrow \C$ be a word function of	degree $k$ corresponding to a set $\inds$. 
	Then for any expander $X$ with an unbiased $G$-labelling,
	\[
	 \abs{\cE_X(f)} ~\leq~ \lambda^{\eta(\inds)}  \cdot |G|^{\frac{k}{2}}\cdot \norm{f}_2\; ~\leq~(2\lambda)^{\floor{\frac{k}{2}}}  \cdot |G|^{\frac{k}{2}}\cdot \norm{f}_2. \]  	
%In particular, we have $\abs{\cE_X(f)} ~\leq~\lambda^{-1}\parens{\lambda |G|}^{k/2}\cdot \norm{f}_2 $
\end{theorem}

%
%In the second half of this section, we consider a subclass  of functions within monomial word functions that we call \emph{monotone word functions}. Essentially, these are word functions for which corrresponding word, $w$ is monotone \ie $w = x_{i_1}\cdots x_{i_k}$ for $i_1\leq \cdots i_k$. We already mentioned that 
%for monomial word functions gets fooled by expander walks upto an exponentially decaying error. However, the error bound has dependence on $|G|$. For monotone word functions we remove this dependence while achieving the same decay in terms of expansion. 

One important case of this class of functions is the group product functions, namely, Boolean valued functions $f$ that take  $x_1,\ldots,x_k \in G$ as
input and output $1$ if and only if the product is equal to some target element $t \in G$. Fooling group product functions is a crucial component
in the construction of some pseudorandom generators for branching programs, \eg \cite{MZ09,D11}. 

We sharpen \cref{theo:word_func_intro} for group product functions by removing the dependence on $|G|$ in the error bound while achieving the same exponential decay in terms of expansion.

\begin{theorem}[Fooling Group Products,~\cref{theo:word_func2}]
   Let $G$ be any finite group, $t \in G$, and  $f(\vec{x}) = \indicator{x_1\cdots x_k = t}$ be a group product. Then for any expander $X$ with an unbiased $G$-labelling,
   \[
   \abs{\cE_X(f)} ~\leq~  (2\lambda)^{k/2}\mper
   \]
\end{theorem}

%\todo{Add general word function}
%
%The key result we show here is that the group product functions (more generally, word functions), as above, have their Fourier support entirely on the maximum possible degree irreps. Moreover, this support is also very sparse (see~\cref{lem:repsame}). Thus, we effectively prove a tight characterization of its Fourier spectrum that enables us to use~ \cref{theo:main_tensor}.  

%As alluded before, the above exactly match the parameters obtained by~\cite{CPT21} in case of $G=\F_2$ in this more general setting. 

%\tnote{CPT used the scalar version for $\Z_2$, golowich badhan did someting else for $\Sigma$ and then state. open problem dependence on $\Sigma$ Our operator-valued result allows us to generalize the approach of cpt to get better bounds. }

%The Fourier basis 

%First, we  Moreover, the above suggests the possibility of working with orthonormal bases other than the Fourier basis for analysis of expander function fooling

%The above theorem, when used for irreps and using some a combinatorial bound from CPT, gives

%\tnote{stuff about how it implies a TV distance bound and improves on Golowich-Vadhan's result. Vanishing bounds for small-norm function}

% structured expander case
\subsubsection{Pseudorandomness of structured expanders}
The above results hold for generic expanders, but since our function is defined on a group, it is natural to consider \enquote{structured} expanders that gel well with the group. In the case of Abelian groups, these are just \textit{Cayley graphs}, using which we obtain a further improvement to~\cref{theo:intro_main_fool}.

\begin{theorem}[Abelian Groups,~\cref{cor:trace,prop:quasi}]
  Let $G$ be an Abelian group and $X$ be a Cayley graph on $G$ and let $\{\chi_j \mid j \in [k]\ \}$ be a set of non-trivial characters of $G$.
  Then for any ordered subset $\inds$ of size $k$,
  \[
  \cE_{X,\inds}(\chi_1\otimes\cdots \otimes\chi_k)~~\leq~~ {\lambda^{\eta(\inds)}} \cdot \indicator{\chi_1\cdots \chi_k = \triv}\mper
  \]
  As a consequence, for every symmetric function $f:\Sigma^n\to \C$ and Cayley expander $X$,
  \[  \abs[\big]{\cE_X\parens[\big]{f}} ~\leq~ O\parens[\Big]{\sqrt{|\Sigma|} \cdot \lambda} \cdot \norm{f}_2\mper
  \]
\end{theorem}

For $G = \Z_2$, this result says that the odd degree characters are perfectly fooled, and thus, every odd function $f:\Z_2^n\to \C$ is perfectly fooled by such a structured expander. This already illustrates the improvement over generic expanders. 

To generalize this to general non-abelian groups, we need to restrict to \emph{class functions}, \ie functions such that $f(gx\inv{g}) = f(x)$ for every $x,g \in G$. Note that for Abelian groups, every function is a class function, as the condition is trivially true due to commutativity. Moreover, we will need a stronger notion of a \enquote{pseudo Cayley graph} for which we omit the formal definition here  (see~\cref{def:pseudo}). The key property of these graphs is that their eigenvectors are given by the Fourier basis functions. 

\paragraph{Tighter Bound for Quasirandom Groups} An often seen phenomenon is that one gets better pseudorandomness properties for groups that are \textit{highly non-abelian}.
One way to quantify this is the notion of a $D$-\textit{quasirandom groups} introduced in a seminal work by Gowers \cite{GowersQR08} which is a group in which the smallest
(non-trivial) irreducible representation (see~\cref{def:irrep}) has dimension $D$. Abelian groups are $1$-quasirandom, whereas on the other extreme, there are matrix groups that are $|G|^{\frac{1}{3}}$-quasirandom (see~\cite{GowersQR08}). We show that such a gain does indeed occur in our setting as well.

\begin{theorem}[General Groups,~\cref{cor:trace,prop:quasi}]\label{theo:intro_quasi} Let $G$ be a $D$-quasirandom group and let $X$ be a \enquote{pseudo Cayley} graph on $G$. Let $\{\chi_j \mid j \in [k]\ \}$ be a set of non-trivial characters of $G$, normalized by their dimension. Then for any ordered subset $\inds$ of size $k$,
\[\cE_{X,\inds}(\chi_1\otimes\cdots \otimes\chi_k)~~\leq~~ {\lambda^{\eta(\inds)}} \cdot {\angles[\big]{\chi_\triv, \,\chi_1\cdots \chi_k}}. \]
As a consequence, for every symmetric class function $f:G^n\to \C$,
\[
	\abs[\big]{\cE_X\parens[\big]{f}} ~\leq~ O\parens[\bigg]{{\frac{\sqrt{|G|}}{\,\,D}} \cdot \lambda} \cdot \norm{f}_2 .
	\]
\end{theorem}

Apart from the quasirandomness factor, the key improvement from~\cref{theo:main_tensor} here is the extra factor of $\angles[\big]{\chi_\triv, \,\chi_1\cdots \chi_k}$. This counts the fractional dimension of the trivial irrep inside the tensor representation $\rho_1\otimes\cdots \otimes\rho_k$. This quantity is much smaller than one, for instance, when $k =2$, it is at most $\frac{1}{D^2}$. Moreover, this quantity can be computed explicitly using basic representation theory, which yields a more precise upper bound. 
%\begin{corollary}
%Let $G$ be a group such that the dimension of its smallest irrep is $D$. Let $f: G^n\to \C$ be a class function that is also symmetric, and $X$ be a $\lambda$-Cayley labelled expander.  Then,
%	\[
%	\abs[\big]{\cE_X\parens[\big]{f}} ~\leq~ O\parens[\Big]{{\frac{|G|^{\frac{1}{2}}}{D}} \lambda} \cdot \norm{f}_2 .
%	\]
%	In particular, for every symmetric function on an Abliean group, we get a bound of $O(\sqrt{|G|}\cdot \lambda)$.
%\end{corollary}

% word function case

%implies that functions with a bounded $L_1$-Fourier norm are efficiently fooled. The key result here is to prove that word functions 
%
%In particular,  $\norm{\widehat{f}_k}_1 \leq b^k$ The key result is 
%More generally, we show that if $h()$

% lower bound 
\subsubsection{Lower Bounds}

We show that our dependence on $\lambda$ in the bound of $\abs{\cE_X(f)}$ in~\cref{theo:intro_main_fool} cannot be improved in general,
no matter the choice of group $G$. Let, $A \subseteq G$  
and $t \in [n]$. We define a symmetric boolean function $\thresh$ as :
 \[\thresh(\vec x) = 1 \;\; \text{ if }\;  |\{i \mid x_i\in A \}| \geq t; \;\; 0 \;\;\text{otherwise}.
\]

\begin{theorem}[Lower Bound for any group]\label{theo:lower_bound_intro}
  Let $G$ be any finite group, and $A \subseteq G$ be any subset such that $\frac{|A|}{|G|} = \frac{1}{2}$. There exists an $\lambda$-expander $X$ such that for every $n$ large enough,
  \[
  \abs[\Big]{\cE_X\parens[\Big]{\threshold{\frac{n+1-\sqrt{n}}{2}}}} ~\geq~  \Omega \parens[\big]{\lambda}.
  \] 
  This lower bound holds even when $X$ is a \enquote{pseudo Cayley graph}	(as in~\cref{theo:intro_quasi}) on $G$.
\end{theorem}

This lower bound places a limitation on how much the quasirandomness of the group or the algebraic structure of the expander can be leveraged in terms of the pseudorandomness of expander walks with respect to symmetric functions. 

Regardless of how ``far'' from Abelian the group $G$ is, a lower bound of $\Omega(\lambda)$ still persists. This lower bound
rules out the possibility of an upper bound of, say, $\frac{\lambda}{D}$ for a $D$-quasirandom group in~\cref{theo:intro_main_fool}. More importantly, it shows that even if one uses an expander with such Cayley-like algebraic structure, one cannot improve the linear dependence on $\lambda$. 

We stress that
proving this lower bound for general finite groups is substantially more challenging than for the $\Z_2^k$ case \cite{CMPPT22}. In general,
this requires the function and the graph to ``interact'' in a non-trivial way, but now, in the presence of (possibly) higher-dimensional
representations, this is substantially more delicate to achieve (see~\cref{subsec:techniques}).

\subsection{Proof techniques}\label{subsec:techniques}

\paragraph{A generalized ''Ignore First Step" Trick.}  To prove our first main result (\cref{theo:main_tensor}), we generalize the technique of~\cite{RR2024} (also, subsequently used in~\cite{RR23}) that introduced a trick that they called \enquote{Ignore First Step} Trick. We generalize this in two significant ways. We first extend it to the setup of general matrix-valued functions. More importantly, we perform a finer analysis to obtain a dependence on $\lambda$ that takes into account the subset of indices $\inds$. This is necessary to yield a bound of $\lambda^{\eta(\inds)}$ as opposed to $\lambda^{\floor{k/2}}$ (even for scalar-valued functions) that would be implied by~\cite{RR2024}.

We give a quick overview of this technique in the very special setup of $\{\pm 1\}$-valued functions that are all identical. We wish to analyze the term:
\[
\Ex{(x_1,\dots, x_n)\sim \rwalk{n}}{f_1(x_{1})\cdots f_k(x_{k})}.
\]
This corresponds to $\inds = \{1,\dots, k\}$. Let us start with $k=2$. This case can be directly handled by the expander mixing lemma, which says that for a $\lambda$-spectral expander,
\[\abs[\Big]{\,\Ex{(x_1, x_2)\sim \rwalk{2}}{f(x_1)\,f(x_2)} \,- \;  \parens[\Big]{\Ex{x\sim \rwalk{1}}{f(x)}}^2  } ~\leq~ \lambda\cdot  \Ex{x\sim \rwalk{1}}{\abs{f(x)}^2}. 
%\quad \text{(Expander Mixing Lemma)}
\]

%\begin{lemma}[Expander mixing lemma]
%	Let $X$ be a $\lambda$-spectral expander and $f,~g:X\rightarrow \C^n$ be two vector valued functions on the vertex set $X$. Let,
%	$\mu_f=\Ex{x\sim X}{f(x)}$ and $\norm{f}^2_2=\Ex{|f(x)|^2}$.  Then we have:
%	\[\abs[\big]{\,\E_{x\sim y}\brackets{\ip{f(x)}{g(y)}    } - \ip{\mu(f)}{\mu(g)} \, } ~\leq~ \lambda\cdot  \norm{f}_2  \norm{g}_2. \]
%\end{lemma}

One interpretation of this lemma is that it reduces the analysis of the mean of the product function over $2$-length walks to the analysis of the mean and variance of the function over a walk of length $1$. The main idea behind the technique is to do such a reduction from a length $k$-walk to analyzing mean and variance over $(k-1)$-length walks.

We do not get into the details of this reduction but explain the trick used to bound such variance terms, the simplest case of which is when $k=3$. For a vertex $x$, let $\rwalk{1}(x)$ be the distribution of $1$-length walks starting from $x$. The term we need to analyze is,
\[
\Ex{x\sim \rwalk{1}}{ \abs[\Big]{\Ex{y\sim\rwalk{1}(x)}{f(x)f(y)}}^2} ~=~ \Ex{x\sim \rwalk{1}}{ \abs{f(x)}^2 {\Ex{y,z\sim\rwalk{1}(x)}{f(y)\overline{f(z)}}}}
\]

The key technical point is the following. The expression on the right formally depends on $x$ but since $f(x)^2 = 1$, this dependence is virtual. More importantly, the distribution on $y,z$ in this expression is the same as sampling $y,z$ independently (of $x$) at distance $2$ in the graph, 
\[
\Ex{x\sim \rwalk{1}}{{\Ex{y,z\sim\rwalk{1}(x)}{f(y)\overline{f(z)}}}} ~=~ {{\Ex{(y,z)\sim \rwalk{2}'}{f(y)\overline{f(z)}}}}.
\]

The right-hand side can now be analyzed by applying the above expander mixing lemma on the graph $X^2$. 
Thus, this trick gets rid of the first variable $x$, and reduces the variance of $2$-length walks to the mean of $1$-length walk (on the squared graph). 

Our proof follows a similar approach, but there are two key complications. One, the functions we need to analyze are matrix-valued, and secondly, the above analysis does not utilize the gaps in the index set $\inds$, and therefore would give a bound akin to~\cite{JMRW22, Ta-Shma17} which is too weak for our purposes.  

Let $\inds = (i_1,\cdots, i_k)$, and let $\Delta_j := i_{j+1}-i_j$ be the $j^{\mathrm{th}}$ gap. To bound the recurrence more tightly, we view the random walk as a sequence of $k$ steps, where the $j^{\mathrm{th}}$ step is on the graph $X^{\Delta_j}$. To implement this approach in the general setup of tensors of operator-valued functions, we introduce auxiliary functions such as,  
\[ 
	g_j(\vec x) ~:=~  \id_{d_1}\otimes\dots \otimes\id_{d_{j-1}}\otimes f_j(x_{i_j})\otimes f_{j+1}(x_{i_{j+1}}) \cdots \otimes f_k(x_{i_k}),  \]

that capture the intermediate state of this random process after $j$ steps. This lets us utilize the large gaps, \ie $|\Delta_j|$ in $\inds$ , to obtain a sharper bound ($\lambda^{\eta(\inds)}$ instead of $\lambda^{\floor{\frac{k}{2}}}$).

\paragraph{Beyond Spectral Gap via structured graphs} The above technique is quite general and works beyond the setup of groups, thereby yielding a general result (\cref{theo:main_tensor}). Moreover, it only uses the fact that $X$ is an expander \ie that it satisfies a spectral gap. While this leads to operator norm bounds, it is not amenable to analyzing trace norms, and one has to appeal to generic bounds such as $\norm{\sfM}_\tr \leq \dim(\sfM)\cdot \norm{\sfM}_\op$ which is suboptimal in many cases. The key insight behind our second main result (\cref{theo:intro_quasi}) is to use additional spectral information (eigenvectors) about the expander $X$, and not just its spectral gap. To do so, we define the notion of \emph{pseudo Cayley graphs}. 

\paragraph{Pseudo Cayley Graphs} These are graphs such that the characters of the group $G$ are its eigenvectors. More precisely, there exists a labeling of its vertices, $\varphi:X\to G$, such that $\chi\circ \varphi$ is an eigenvector of the graph adjacency matrix $A_X$ for every character $\chi$ of $G$. Note that this property is true for Cayley graphs over Abelian groups. Moreover, one can also build examples over non-Abelian groups (see \cref{example:pseudo1,example:pseudo2})  

 To make use of the above structure, we use a key fact from representation theory, which says that the product of characters over any finite group $G$ can be decomposed a linear sum,
	\[
	 \chi_{\alpha}(g)\cdot \chi_{\beta}(g) ~=~ \sum_{\gamma} c^{\alpha,\beta}_\gamma \cdot  \chi_{\gamma}(g).
	 \]
These coefficients are called Clebsch--Gordan coefficients for $G$. Therefore, our expression can be inductively unrolled by alternating the operations-- (i) taking a step of the walk (which can be handled now that characters are eigenvectors), and (ii) decomposing the product of characters as a linear sum. This leads to a precise calculation of the mean over random walks (see~\cref{theo:quasi}) as opposed to an upper bound for the operator norm.

\paragraph{Lower Bound} This precise calculation comes in handy not just to prove the sharper bound in~\cref{theo:intro_quasi}, but also for the lower bound. The candidate hard function is a generalization of the Boolean threshold function which was used in the analysis of~\cite{CPT21}. However, their construction of the graph is specific to $\Z_2$ and does not generalize to other groups (even $\Z_p$). Moreover, in Abelian groups the representations are $1$-dimensional irreps and thus, $\abs{\tr(\sfM)} = \norm{\sfM}_\op =  \norm{\sfM}_\tr$. However, in higher dimensions even if $\norm{\sfM}_\op \geq \lambda$, the trace, $\tr(\sfM)$ can be zero which is actually the quantity which we need to lower bound. To tackle this, we compute this trace exactly at level $2$ (\cref{corr:level2_proj}) and combine it with the precise computation of the mean for pseudo Cayley graphs (see~\cref{theo:quasi}).
%We follow a completely different approach  

%\begin{fact}[Decomposition of tensor representations]
%	\label{fact:sumChar}
%	Let, $\alpha,\beta\in \widehat{G}$ be two irreps of 
%	a finite group $G$. There exists a change of basis transformation $\mathsf{N}_{\alpha,\beta}$, and non-negative integer coefficients $\{ c^{\alpha, \beta}_\gamma \mid \gamma \in\widehat{G}\}$ such that for any $g\in G$: 
	 
%	\begin{enumerate}[label=(\alph*)]
%		\item If $\oplus_{\gamma \in \widehat{G}} c^{\alpha,\beta}_\gamma \gamma$ is  the decomposition of $\alpha\otimes \beta$ as a $G$-representation.  Then, for all $g\in G$ it holds that: \[ \chi_{\alpha}(g)\cdot \chi_{\beta}(g) ~=~ \sum_{\gamma} c^{\alpha,\beta}_\gamma \cdot  \chi_{\gamma}(g). \]
		%\item  If $\oplus \delta_{k}$ is  the decomposition of $\beta$ into irreps. as a $G$-representation. Then it holds that,
		%\[   \E_{g}\brackets{\chi_{\alpha}(g)\cdot \chi_{\beta}(g)}= \sum_{k}    \]		
%\end{fact}

\section{Preliminaries}
\label{sec:prelim}
\subsection{Random walks on expander graphs}
Throughout the paper, $X=(V,E)$ will be a $d$-regular $\lambda$-expander graph. 
We write $A_X$  to denote the degree normalized adjacency operator of $X$.

%We first recall the definition of $\lambda$-expander graphs. 

\begin{definition}[$(d,\lambda)$-expander]  A graph $d$-regular graph $X=(V,E)$ is called $(d,\lambda)$-expander if $\max\{|\lambda_2|,|\lambda_{N}|\} \leq \lambda$ where  $\lambda_1\leq \lambda_2\leq\dots \leq \lambda_{N}$ are the eigenvalues of $A_{X}.$  
\end{definition}

 We write $x\sim y$ denote sampling of an edge, $(x,y)$ from $X$. A key tool in analyzing expanders is the expander mixing lemma:
%The famous \emph {expander mixing lemma} that the average value of the product of two functions over the edges in $X$ is roughly equal to the product of their means. We have the following formal statement. 
 \begin{lemma}[Expander mixing lemma]
	\label{lem: EML}
	Let $X$ be a $\lambda$-spectral expander and $f,~g:X\rightarrow \C^n$ be two vector-valued functions on the vertex set $X$. Let,
	$\mu_f=\Ex{x\sim X}{f(x)}$ and $\norm{f}^2_2=\Ex{|f(x)|^2}$.  Then we have:
	\[\abs[\big]{\,\E_{x\sim y}\brackets{\ip{f(x)}{g(y)}    } - \ip{\mu_f}{\mu_g} \, } ~\leq~ \lambda\cdot  \norm{f}_2  \norm{g}_2. \]
	%where $\mu_f~=~\E[f]\in \C^n$, $\sigma^2_f= \E_{x\sim X}\brackets{ \norm{f(x)}^2} - \norm{\mu_f}^2$ and $\mu_g,\sigma_g$ are similarly defined. 
\end{lemma}

\paragraph{Random walk notation} We find it helpful to define a few shortands associated with random walks on $X=(V,E)$ to streamline our presentation. We list them below.
\begin{itemize}
\item We write $''\vec x\sim \rwalk{n}''$ to denote uniform sampling of an $(n-1)$-step (or, $n$ vertices long) random walk, $(x_1,x_2,\dots, x_n)=:\vec x$ on $X$. 
\item  Given, $x\in V$, the notation $''\vec x\sim \rwalk{n}(x) '' $ denotes 
uniform sampling of an $(n-1)$-step (or, $n$ vertices long) random walk, $\vec x$  conditioned on $x_1=x$. 
\item  The expression ''$x\sim^k y$''  denotes sampling a pair, $(x,y)$, of vertices from $X$ that are at a distance of $k$. 
 
%\item \snote{one more notation?}
\end{itemize}

\begin{fact}[Distribution for a single]\label{fact:single}
	Fix any $k \in [n]$. Then, the marginal distribution on $x_k$ when $\vec{x}\sim \rwalk{n}$ is uniform over $X$.
\end{fact}

\subsection{The main quantity}
Let $G$ be any finite group and $X=(V,E)$ be an expander graph. 	A $G$-labeling (or, simply labeling), $\phi$,  of $X$ is  a map $\phi: X\rightarrow G$. Given any such labeling $\phi$, we say $\phi$, is \emph{unbiased} if 
	\[\Pr{x\sim X}{\phi(x)=g}={|G|^{-1}} \text{ for all }g\in G\]
In this work, our focus is functions of the form $f:G^n\rightarrow \C$. We will always assume that the labelling is unbiased and use $f(x)$ to denote $f\circ \phi$ to prevent clutter. 
\subsection{Inner products and norms}

Let $\C^d$ be the $d$-dimensional complex inner product space equipped with the inner product. 
%$\cL(\cH)$ the algebra of bounded linear operators\footnote{For most
% applications, one can think of $\cH = \C^n$ for some $n$, and
 % $\cL(\cH) = \matr M_n(\C)$, the space of $n\times n$ complex
 % matrices.} on $\cH$. 
 We denote by $\textup{U}_d$ the
 group of $d$-dimensional unitary matrices. Let $A,B\in \C^{d\times d}$ be two complex matrices. We have the following inner products and norms:
\begin{itemize}
\item $\ip{u}{v}:=\Ex{i\sim [d]}{u_iv_i^*}$
	 \item $\ip{A}{B}_{\HS}:=\tr\parens{A^*B}=\tr\parens{B^*A}$
	 	 \item $\norm{A}^2_{\HS}=\tr\parens{A^*A}=\sum _{i,j} |A_{i,j}|^2$
	 	 \item $\norm{A}_{\tr}=\tr\parens{ \sqrt{A^*A}}$
	 	 \item $\norm{A}_{\op}=\sup_{~\norm{x}=1} ~ \norm{Ax}$ where $\norm{\cdot}$ denotes the norm associated with $\C^d$. 
\end{itemize}
%We have the following standard fact that can be derived by combining Von Neumann's trace inequality with H\"{o}lder's inequality. 

%\begin{fact}
%	For any matrix $\sfM \in \C^{d\times d}$ and  vectors $u,v \in \C^d$, $\ip{\sfM u}{\sfM v} ~\leq~ \opnorm{\sfM}^2\cdot\ip{u}{v}$.
%\end{fact}

\subsection{Fourier Analysis on Finite Groups}
We always use $G$ to denote an arbitrary finite group (not necessarily abelian) unless specified otherwise. Denote by $L^2(G) = \{ f: G\to \C \}$, the space of complex-valued functions equipped with the following inner product, 
\begin{equation*}
	\ip{f}{g}  = \Ex{x\sim G}{\parens[\big]{g(x)^*f(x)}}.
\end{equation*}
This induces the norm is $\norm{f}^2 =  \E_{x\sim G}\brackets[\big]{|f(x)|^2}$. 

\paragraph{Group Representations}
We will use the
notion of a group representation\footnote{Additional background on
  representation theory of finite groups can be found
  in~\cite{serre96}.}. \textit{Weyl's unitary trick}, says that for a
large family of groups (which includes all finite groups), every
representation can be made unitary and thus, we can restrict to
studying these.

\begin{definition}[Irreducible Group Representation]\label{def:irrep}
Let $G$ be a finite group.  A unitary representation of $G$ is a group homomorphism $ \rho: G\rightarrow \textup{U}_d$ for some $d$, \ie $\rho(g_1g_2) = \rho(g_1)\rho(g_2)$ for every $g_1,g_2 \in G$. The character, $\chi_{\rho}: G\rightarrow \C$ associated with $\rho$ is  the function: $\chi_\rho=\tr\ccirc \rho$. Note, that characters are not necessarily homomorphisms. A representation is called \textit{irreducible} (or irrep) if there exists no subspace of $V\subseteq \C^d$ such that $\rho(g) V\subseteq V$ for all $g \in G$. The set of irreps of $G$ is denoted as $\widehat{G}$.
\end{definition}

%
%\todo{incomplete fact}
%\begin{fact}[Fourier Basis] For every finite group $G$, the set of functions $$
%	
%\end{fact}

%\noindent Every finite group has two special representations,  which are,
%\begin{enumerate}[topsep=-0.1em,leftmargin=*]
%	\item (Trivial representation ) - $(\rho, \C)$ where for every $g$,  $\rho(g) = 1$.
%	\item ((left) Regular representation ) - $(\rho_{\reg},  \mathcal{V}_{\reg})$ where, $ \mathcal{V}_{\reg}= \C[G]$ is
%              a vector space with the elements of $G$ being an orthonormal basis,  and  $\rho_{\mathrm{reg}}(g):  h \mapsto g\cdot h$.  
%\end{enumerate}
%
%\begin{fact}\label{lem:semisimple}
%	Let $G$ be a finite group and let $\mathcal{V}_{\reg}$ be the regular representation over $\mathbb{C}$. We have
%	\begin{align*}
%		\mathcal{V}_{\reg} ~=~ \C[G] ~\cong  \bigoplus_{ (\rho, V_\rho) \in \irrep{G}} \textup{dim}(\rho) \cdot \mathcal{V}_{\rho},
%	\end{align*}
%	where $\irrep{G}$ denotes the set of irreducible unitary representations of $G$.
%\end{fact}

%\begin{definition}
%	Let, $\rho:G\to \textup{U}_d $ be a representation of a finite group $G$. The character, $\chi_{\rho}: G\rightarrow \C$ associated with $\rho$ is defined as the function: $\chi_\rho(g)=\tr(\rho(g))$ for all $g\in G$. Note, that characters are not necessarily homomorphisms. 
%\end{definition}
When $G$ is abelian, all irreducible representations are one-dimensional. Thus, in this case, the set of characters, and the set of irreps coincide. Moreover, for abelian $G$, the set of characters form an orthogonal basis of $\C[G]$. This does not hold for arbitrary finite groups $G$. Nevertheless, even for arbitary finite groups, the set characters do satisfy the orthogonality conditions.
\begin{fact}\label{lem:mean_zero}
	Let, $\rho, \gamma \in \irrep{G}$ be two irreducible representations. Then,
	\[
	\ip{\chi_{\rho}}{\chi_\gamma}=
	\begin{cases}
		d_\rho, & \text{if } \rho=\gamma, \\
		0, & \text{otherwise}.
	\end{cases}
	\]
Moreover, for any non-trivial representation $\rho$ of any finite group $G$, $\E_{g\in G} \rho(g) = 0$.	 
\end{fact}
%
%\begin{lemma}
%	
%\end{lemma}
%\begin{proof}
%	Follows from orthogonality of the Fourier basis.
%\end{proof}

%\begin{lemma}\label{lem:tensor}
%	Every representation of $G$ can be decomposed as a sum of irreducibles. 
%Moreover, $\rho_{\mathrm{triv}} \in \rho\otimes\psi$ if and only if $\psi =\rho^*$. And the number of copies of the trivial representation is exactly one.
%\end{lemma}
%Given this equivalence, we will find it convenient to work with the
%operator norm version referred to as \emph{bias} in the
%literature~\cite{CMR13}.

%\begin{definition}[Biased Distribution on $G$]\label{def:bias_group}
  %
  %Let $\epsilon \in (0,1)$.
  %
  %We say that a multiset $S$ of elements of a group $G$ is
  %$\epsilon$-biased if for every non-trivial irreducible
  %representation $\rho$, we have $\opnorm{ \Ex{s\sim S}{\rho(s)} }
  %\leq \epsilon$. We sometimes use the shorthand $\bias(S) \leq
  %\epsilon$, where $\bias(S) = \lambda(\Cay(G,S))x$.
%\end{definition}
\noindent

\noindent

%\paragraph{Cayley graphs on groups}
%For a multiset $S \subseteq
%G$, $\Cay(G,S)$ denotes a multigraph\footnote{Note that unless $S =
%	S^{-1}$, the graph $\Cay(G,S)$ is a directed multigraph.} with the
%vertex set being $G$ and edges $\set{(g,sg)\;\vert\; g\in G,\; s \in S
%	}$.

%\begin{definition}[Quasi-Abelian Cayley Graphs]
%A graph $X= \Cay(G,S)$ is called a Quasi-Abelian Cayley Graph if $S$ is a union of conjugacy classes of $G$. Moreover, if $X$ is such a graph then, every orthogonal basis vector, $\rho_{i,j}$ is an eigenvector of $X$ with an eigenvalue $\lambda_\rho := \frac{1}{d_\rho} \sum_s \chi_\rho(s)$.
%\end{definition}
	
%Irreducible representations of Abelian groups, called
%\emph{characters}, have dimension $1$. Thus, this definition coincides
%with the usual one of $\epsilon$-biased distribution \emph{fooling}
%non-trivial characters~\cite{NN90,AGHP92}. These pseudorandom
%distributions were introduced in the pioneering work of Naor and Naor
%where several applications to derandomization were given~\cite{NN90}.

\subsection{Complex valued functions on groups}

 For every finite group $G$, the set of functions given by the matrix entries of the irreps, \ie $\{\rho_{ij} \mid \rho \in \widehat{G}, i, j \in [d_\rho]  \}$, form an orthogonal basis for the space of all functions, $L^2(G)$.

%For a function $f$, we denote its \textit{adjoint} by $\tilde{f}(x) := f(x^{-1})^*$. The operation of convolution generalizes as, 
%\[(f*g)(x) :=\E_{y\sim G}\brackets[\big]{f(xy^{-1})g(y)} = \E_{y\sim G}\brackets[\big]{\tilde{f}(y)^* g(yx)}.
%\] 

\begin{definition}[Fourier Coefficient]
	For any irrep $\rho$, we have $\widehat{f}(\rho):=\E_{x}\brackets[\big]{f(x)\cdot \rho(x)}.$ The Fourier coefficient of the trivial irrep as  $\mu(f) := \widehat{f}(\rho_\triv)$.
\end{definition}

\begin{fact}
	\label{fact:fourier}
	The following identities hold for the Fourier transform,
	\begin{enumerate}
		\item {\bf (Fourier inversion)}
		$\;\; f(x)= \sum_{\rho \in \widehat{G}}$ $d_\rho\angles[\big]{\,\widehat{f}(\rho), \, \rho(x)}$ .
		
		\item {\bf (Plancharel's identity)} $\;\;\norm{f}^2 =\sum_{\rho \in \widehat{G}}d_{\rho}{\norm{\widehat{f}(\rho)}^2_{\HS}}$.
	\end{enumerate}
\end{fact}

\paragraph{Product Groups}
In this paper, we will work with product groups. The following fact characterizes the irreducible representations of $G^n$ in terms of irreps of $G$. 
\begin{fact}
	$\widehat{{G^n}} = \{ \rho_1\otimes \cdots\otimes \rho_n\mid \rho_i \in \irrep{G} \}$. We use $\vrho_T$ to represent $\rho_1\otimes \cdots\otimes \rho_n$ such that $T = \{i \mid \rho_i 
\neq \triv \}$. Moreover, $|\vrho| :=|T|$. 
\end{fact}

\begin{definition}[Degree Decomposition]
	For $f: G^n\to \C$, we use $f_k$ to denote the function corresponding to its $k^{\mathrm{th}}$-level, \ie $f_k(\vec{x}) = \sum_{\vrho, |\vrho| = k} d_{\vrho}\ip{\hat{f}(\vrho)}{\vrho(\vec{x})}$. In the Boolean case ($\Z_2^n$), this is also referred to as the degree $k$ component of $f$.
\end{definition}

%\begin{claim}\label{lem:rhoij}
%Let $\rho_1, \rho_2 \in \Irrep(G)$. Then, 
%$\Ex{g \in G}{\rho_1(g)\otimes \rho_2( g) } \neq 0 $	only if $\rho_1 = \rho_2^*$. 
%\end{claim}
%\begin{proof}
%	Let $\rho_1\otimes\rho_2 = \oplus_{\rho \in \Irrep(G)} a_\rho \rho $. Thus, one can block-diagonalize the space with $a_\rho$ copies of $\rho$. Now, since $\sum_g \rho(g) = 0$ for any non-trivial irrep $\rho$, this term is non-zero iff the trivial representation has a non-zero multiplicity. This happens iff $\rho_i = \rho_j^*$. Moreover, the multiplicity is one.
%%	 $\mathrm{triv} \subseteq \rho_1\otimes\rho_2 $ 	
%\end{proof}

%\todo{Introduce Fourier stuff and inner products}

%\subsection{Notation}
%\tnote{Should I use the degree word for $f_k$?}

%\begin{itemize}
	
%\end{itemize}

\section{Expander Walks and Product functions} 
\label{subsec: foolingtensor}

In this section, we will prove the main claim about the fooling of tensored product matrix-valued functions. This can be seen as the matrix-valued generalization of~\cite{CPT21}. We will then apply it to our Fourier basis elements, \ie irreps which are tensor product functions, to obtain our main result for general functions.  To state our theorem, we will need a few pieces of notations borrowed from~\cite{CPT21}  that we describe below. 

\paragraph{Notation} 
Let $\inds  = \set{i_1 < i_2<\cdots < i_{k-1}< i_k}$ be an ordered subset of $\{1,2,\dots,n\}$. We define the following key quantities: 
\begin{itemize}
	\item $	\mathcal{I}_k ~=~ \braces[\big]{\{1, k-1\} \subseteq I  \subseteq [k-1] \mid  \;\forall\; 1 < j < k-1,\; \{j,j+1\} \cap I \neq \emptyset}$. 
	\item $\Delta_{j}(\inds) ~=~ i_{j+1}-i_{j}$. 
	%\item $\Delta(\inds) ~=~ \sum^{k-2}_j\min(\Delta_j,\Delta_{j+1})$ 
\end{itemize}

%The key quantity we wish to analyze in this section is,
%\[
%\cE_{X,\inds}(f_1\otimes \cdots \otimes f_k) = \Ex{\vec{x} \sim \rwalk{n}}{f_1(x_{i_1}),\dots, f_k(x_{i_k}))} 
%\]
In this section, we can state our main theorem that we prove in this section.

\begin{theorem}\label{theo:tensor_fool}
	Let, $X$ be an $\lambda$-expander graph and let $\inds  = \set{i_1 < i_2<\cdots < i_{k-1}< i_k}$ be an ordered subset of $[n]$. Let $\{f_j: X\rightarrow  \matr M_{d_j}(\C) \mid j \in [k]\ \}$ be set of matrix valued functions such that  ${\E_{x\sim X}[f_j(x)]}=0$, and $\max_x\norm{f_j(x)}_{\op}\leq 1$. Then,
	\[ 
	\norm[\big]{\E_{\vec{x}\in \rwalk{n}} [f_1(x_{i_1})\otimes\cdots \otimes f_k(x_{i_k}) ]}_\op 	~\leq~\sum_{I\in \mathcal{I}(k)} \lambda^{\sum_{i\in I}\Delta_i(\inds)} .  
	\]
\end{theorem}

\begin{proof}
For any $j\in [k]$,  we define the following shorthands: 
\begin{align}
	g_j(\vec x) ~&:=~  \id_{d_1}\otimes\dots \otimes\id_{d_{j-1}}\otimes f_j(x_{i_j})\otimes f_{j+1}(x_{i_{j+1}}) \cdots \otimes f_k(x_{i_k}).\nonumber\\
	h_j(\vec x) ~&:=~  \id_{d_1}\otimes\dots \otimes\id_{d_{j-1}}\otimes f_j(x_{i_j})\otimes \id_{d_{j+1}} \otimes\cdots \otimes \id_{d_{j+1}}.\nonumber\\
\label{eq:g_h}\text{And thus, for } j <k, \quad	g_j(\vec x) ~&=~ h_j g_{j+1} \; \;.\\
 \label{eq:average_matrix} \avgmatr_j (y) ~&=~ \Ex{\vec x\sim X_{n-i_j}(y)}{{g}_j(y, x_{i_{j+1}}, \cdots , x_{i_k})}, \;\;  \\
 \avgmatr_k (y) ~&=~ g_k(y) . \;\; \nonumber 
\end{align}

Equation~\cref{eq:g_h} follows from the mixed product property of tensors. The last equation here denotes expectation under the distribution in which the random walk starts from a fixed point $y$, \ie $x_{i_j} = y$ and then continues for $n-i_j$ steps until it reaches $x_n$. We start by defining two key quantities that we will work with throughout.  	
\begin{align*}
	\ep_{{j}} ~&:=~\norm[\big]{\,\E_{\vec{x}\sim \rwalk{n}}\brackets{g_j(\vec x)}\,}_\op = \norm[\big]{\,\E_{y\sim X}\brackets{\avgmatr_j(y)}\,}_\op, \\[0.5em]
	\zeta^2_{j} ~&:=~ \E_{y\sim X}\brackets[\Big]{ \norm[\big]{\avgmatr_j(y)}_\op^2 }.
\end{align*}	
%\[\ep_{{j}}:=~\norm[\big]{\,\E_{\vec{x}\sim X^n}\brackets{g_j(\vec x)}\,}_\op, \;\; \zeta^2_{j} ~:=~ \E_{x\sim X}\brackets[\Big]{ \norm[\big]{\E_{\vec x\sim X_{n-i_j +1}(x)} \brackets{{g}_j} }_\op^2 }. 
%	\]
The first equation above has an equality as the distribution that draws a random walk $\vec{x}\sim X^{n}$ and outputs the last $n-i_{j}+1$ co-ordinates of the walk, is equivalent to the distribution that outputs a random sample from $X^{n-i_{j}+1}$. Note that our goal is to upper bound~$\ep_1$, and we recursively bound this by using the technique from~\cite{RR2024}. The matrix $\avgmatr_j$ essentially denotes the averaging that has been done for the function after $f_j$.  \\
%	Towards this we also define a variance quantity $\zeta_{j}$ as follows: 
%	\[ \ep_{{1}}=\norm[\big]{\E_{\vec{x}\sim X^n}\brackets{\vec{f}_{i_1:n}(\vec x)}}_\op ~=~
%	\norm[\big]{\E_{\vec{x}\in X^n} [\otimes^n_{i=1}  f_i(x_i) ]}_\op 
%	\]
%	\[  
%	\zeta^2_{j} ~:=~ \E_{x\sim X}\brackets[\bigg]{ \norm[\Big]{\E_{\vec x\sim X_{n-i_j +1}(x)} \brackets{\vec{f}_{i_j:n}} }_\op^2 }
%	\]

	\textbf{Base cases} ~  Since we will bound it recursively starting from $j=k$, that forms the base case. By using the assumptions on $f_k$, we have:
	\begin{equation}
\label{eq:base_ep}			\ep_k ~=~ \norm[\big]{\Ex{y\sim X}{f_k(y)}}_\op = 0, \;\; 	\zeta_k^2 ~=~ \Ex{x}{\norm{f_k(x)}_\op^2} ~\leq~ 1  \; .
		\end{equation}

	\textbf{Bounding $\ep_{i_{j}}$.} 
%We begin with the observation that ${g}_{j}$ only depends on the last $n-i_{j}+1$ co-ordinates of any input. Moreover, the distribution that draws a random walk $\vec{x}\sim X^{n}$ and outputs the last $n-i_{j}+1$ co-ordinates of the walk, is equivalent to the distribution that outputs a random sample from $X^{n-i_{j}+1}$. Combining these two observations we get:
%	\begin{align}
%		\label{eq: shortwalk}
%		\opnorm{\,\,\E_{\vec{x}\sim \rwalk{n}}\brackets{g_j(\vec x) } \,   } = 
%		\opnorm{\E_{\vec{x}\sim \rwalk{n-i_j+1}}\brackets{g_j(\vec x)   }    } .
%	\end{align}
%	where given $\vec x\in X^{n-i_j+1}$, the notation 
%	\[\vec{f}_{i_j:n}(\vec x):= \id_{d_1}\otimes\dots \otimes\id_{d_{i_j-1}}\otimes f_{i_j}(x_1)\otimes\cdots\otimes f_{n}(x_n).\] 
	We will use the variational definition of the operator norm. To start let $u, v$ be any unit vectors in $\C^{d_j}\otimes \cdots \C^{d_k}$. For any $y\in X$, we have:
	%We define $\vec{\rho}_{I}=
	%Let $i_1<i_2<\dots<i_k$ be the ordering of the co-ordinates in $S$. 
	\begin{align*}
		\ip{u}{N_j(y), v}~&=~ {\; \E_{\vec{x}\sim X^{n-i_j+1}(y)}\;\ip{\,u}{\;g_j(y,\vec x)\,v\,}\;} 
		\\[0.5em]
		~&=~ {\; \E_{\vec{x}\sim X^{n-i_j+1}(y)}\;\ip{\,h_j(y)^*\,u}{\;g_{j+1}(\vec x)\,v\,}\;}
		&&[\text{By \cref{eq:g_h}}] \\[0.5em]
		~&=~ {\; \E_{y\sim^{\Delta_j} x_{i_{j+1}}}\angles[\bigg]{\,h_j(y)^*\,u, \; \Ex{\vec x\sim X_{n-i_{j+1}}(x_{i_{j+1}})}{g_{j+1}(x_{i_{j+1}},\vec x)}\,v\,}\;} \\[0.5em]
			~&=~ {\; \E_{y\sim^{\Delta_j} x_{i_{j+1}}}\;\ip{\,h_j(y)^*\,u}{\;\avgmatr_{j+1}(x_{i_{j+1}})\,v\,}\;}.
%				~&=~ \sup_{u,v}\;\abs[\Big]{\; \E_{\vec{x}\sim X^{n-i_j+1}}\ip{\,h_j(x_{i_j})^*\,u}{\;h_{j+1}(x_{i_{j+1}})\, g_{j+2}(x_{i_{j+2}}, \ldots, x_{i_k})\,v\,}\;} &&[\text{By \cref{eq:g_h}}] \\
%						~&=~ \sup_{u,v}\;\abs[\Big]{\; \E_{\vec{x}\sim X^{n-i_j+1}}\ip{\,h_j(x_{i_j})^*\,u}{\;h_{j+1}(x_{i_{j+1}})\, \Ex{\vec{y}\sim \rwalk{}(x_{i_{j+1}})}{g_{j+2}(\vec{y})}\,v\,}\;} &&[\text{By \cref{eq:g_h}}] .
	\end{align*}
	The main step that happened above was that the sampling of $\vec{x}$ was broken into two steps, (i) sampling  $x_{i_{j+1}}$ after a walk of length $\Delta_j$, and (ii) the remaining $\vec{x}$ starting from this $x_{i_{j+1}}$. The key trick now is to interpret this term as a function of $x_{i_j}, y$ which are neighbouring vertices when sampled from $X^{\Delta_j}$. Therefore, we can apply EML (\cref{lem: EML}). Note that $\E_y [h_j^*u] = I\otimes\E_y[ f_j^*u]\otimes I = 0 $ by our assumption on $f_j$. Similarly, $\max_y \norm{h_j(y)}_\op^2 \leq 1$. Thus, 
	\begin{align}
\abs[\Big]{\E_{y\sim^{\Delta_j} x_{i_{j+1}}}\;\ip{\,h_j(y)^*\,u}{\;\avgmatr_{j+1}(x_{i_{j+1}})\,v\,} }^2 ~&\leq~ \lambda^{2\Delta_j} \cdot \Ex{y\sim X}{\norm{h_j^*(y)u}^2}\cdot \Ex{z\sim X}{\norm{\avgmatr_{j+1}(z)v}^2} 	\nonumber\\
	 ~&\leq~ \lambda^{2\Delta_j} \max_y \norm{h_j(y)}_\op^2 \cdot \zeta_{j+1}^2  .\nonumber\\
	 \implies\;\;\;\; \sup_{u,v} \;\abs[\big]{\ip{u}{N_j(y), v}} ~&\leq~ \lambda^{\Delta_j} \cdot \zeta_{j+1}, \nonumber\\
 	\label{eq: ep}  \ep_j ~&\leq~  \lambda^{\Delta_j} \cdot \zeta_{j+1}.
	\end{align}
	
%	Since the above holds for every 
%	 and Cauchy-Schwarz on the last expression and using the fact $\opnorm{\E_{x}[f_{i_j}]}= 0$, we get:
%
%	\begin{align*}
%				~&=~ \abs[\Big]{  \E_{x\sim^{i_{j+1}-i_j} x'}\angles[\big]{\parens[\big]{\vec{f}_{i_j}(x)}^* u\,, ~\E_{\vec{x'}\sim X^{n-i}(x')}[ \vec{f}_{i_{j+1}:n}(\vec{x'})]v} }, && \; .
%	\end{align*}
%	\begin{align}
%		\label{eq: ep}
%		%\ep_{i_{j}}~\leq~\mu_0\cdot \ep_{i_{j+1}} + \lambda^{\Delta_{j}} \cdot \zeta_{i_{j+1}}.
%		\ep_{{j}}~\leq~ \lambda^{\Delta_{j}} \cdot \zeta_{{j+1}}.
%	\end{align}
%	
	\textbf{Bounding $\zeta_{{j}}$.}
	To bound this quantity, we will again bound the operator norm by bounding the quadratic form for a fixed unit vector $v$, 
	\[
	\zeta_{j}^2~=~ \E_{y\sim X}\brackets[\Big]{\, \norm[\big]{\avgmatr_j(y)}_\op^2 } 
		~=~ \sup_{v, \norm{v}=1}\; \E_{y\sim X}\brackets[\Big]{\, \angles[\big]{\,\avgmatr_j(y)\, v, \; \avgmatr_j(y)\, v} },
	\]
	 and then apply EML (\cref{lem: EML}) to get a recursive upper bound. Fix a $j\in [k]$, and a unit vector $v \in \C^{d_j}\otimes \cdots \C^{d_k}$. We have: 
	\begin{align*}
	 \E_{y\sim X}\brackets[\Big]{\, \angles[\big]{\,\avgmatr_j(y)\, v, \; \avgmatr_j(y)\, v} } ~&=~ \Ex{\substack{y\sim X,\\ \vec{x}, \vec{z} \,\sim \rwalk{n-i_j}(y) }}{\, \angles[\big]{\,g_j(y, \vec{x})\, v, \; g_j(y, \vec{z})\, v} }   \\
	 ~&=~ \Ex{\substack{y\sim X,\\ \vec{x}, \vec{z} \,\sim \rwalk{n-i_j}(y) }}{\, \angles[\big]{\,h_j(y) g_{j+1}(\vec{x})\, v, \; h_{j}(y)g_{j+1}(\vec{z})\, v} }   \\
	 ~&\leq~ \Ex{\substack{y\sim X,\\ \vec{x}, \vec{z} \,\sim \rwalk{n-i_j}(y) }}{\, \norm{h_j(y)}_\op^2 \angles[\big]{\, g_{j+1}(\vec{x})\, v, \; g_{j+1}(\vec{z})\, v} }   \\
	  ~&\leq~ \Ex{\substack{y\sim X,\\ \vec{x}, \vec{z} \,\sim \rwalk{n-i_j}(y) }}{\, \angles[\big]{\, g_{j+1}(\vec{x})\, v, \; g_{j+1}(\vec{z})\, v} }.    %		~&=~\E_{x\sim X}\Big[\ip{   \E_{(x,\vec{y})\sim X^{n-i_j+1}(x)}[ \vec{f}_{i_j:n}(x,\vec{y})]v    }{\;    \E_{(x,\vec{z})\sim X^{n-i+1}(x)}[ \vec{f}_{i_j:n}(x,\vec{z})]v    } \, .
	\end{align*}
	The first inequality uses  $\ip{h_j(y) u}{h_j(y) u} ~\leq~ \opnorm{h_j(y)}^2\cdot\ip{u}{u}$ for $u = \Ex{\vec{x}\sim \rwalk{n-i_j}(y)} {g_{j+1}(\vec{x})\, v}$. The last inequality follows from our assumption on $f_j$ as $(\norm{h_j(y)}_\op^2 = \norm{f_j(y)}_\op^2 \leq 1)$. The above random process of sampling $y$ and then two paths $\vec{x}, \vec{z}$ from it is equivalent to sampling the starting points $x, z$ such that are a distance $2 \Delta_j$ apart, and then performing a random walk from each giving rise to the paths $(\vec{x}, \vec{z})$. 
		\begin{align*}
	 \E_{y\sim X}\brackets[\Big]{\, \angles[\big]{\,\avgmatr_j(y)\, v, \; \avgmatr_j(y)\, v} }	~&\leq~ \Ex{{x\sim^{2\Delta_j} z}}{\, \angles[\big]{\, \E_{\vec{x}\sim \rwalk{n-i_{j+1}}(x)} g_{j+1}(\vec{x})\, v, \; \E_{\vec{z}\sim X^{n-i_{j+1}}(z)} g_{j+1}(\vec{z})\, v} }
\\[0.5em]
~&=~ \Ex{{x\sim^{2\Delta_j} z}}{\, \angles[\big]{\, \avgmatr_{j+1}(x)\, v, \; \avgmatr_{j+1}(z)\, v} }
\\[0.5em]	~&\leq~ \norm[\big]{\Ex{x}{\avgmatr_{j+1}(x)}v\,}^2 +\lambda^{2\Delta_{j}}\cdot \Ex{x}{\norm{\avgmatr_{j+1}(x)\,v}^2}, && \hspace{-4em}(\text{EML, \cref{lem: EML}})\\[0.5em]
%~&\leq~\E_{y\sim^{i_j} z}\Big[\ip{   \E_{\vec{y}\sim X^{n-i}(y)}[ \vec{f}_{i+1:n}(\vec{x})]v    }{    \E_{\vec{y}\sim X^{n-i}(y)}[ \vec{f}_{i+1:n}(\vec{y})]v    }
%		\Big] 
	~&\leq~ \ep^2_{{j+1}} +\lambda^{2\Delta_{j}}\cdot \zeta^2_{{j+1}}. &&\hspace{-5em}(\text{Definition of $\ep_{j+1}, \zeta_{j+1}$})
	\end{align*}
Since the above holds for every $v$, we can deduce: 
	\begin{align}
			\zeta_{j}^2 ~&\leq~ \ep_{{j+1}}^2 +\lambda^{2\Delta_{j}}\cdot \zeta_{{j+1}}^2, \nonumber\\
		\label{eq: zeta} 
	\implies\;\;\;	\zeta_{j} ~&\leq~ \ep_{{j+1}} +\lambda^{\Delta_{j}}\cdot \zeta_{{j+1}}.
	\end{align}
	\textbf{Bounding $\ep_{{1}}$.}
	Combing the equations~\cref{eq: ep} and~\cref{eq: zeta},   we get: 
	\[\zeta_{j}~\leq~ \lambda^{\Delta_{{j+1}}}\zeta_{{j+2}}+\lambda^{\Delta_{{j}}}\zeta_{{j+1}}.\]  
	
	We claim for $j=2,\dots k-1$, it holds that:
	\begin{align}
		\label{eq: inductive}
		\zeta_{j} \leq {\sum}_{T\in \mathcal{T}_j} \;\lambda^{\sum_{t\in T}\Delta_t},
	\end{align}
	where  $\mathcal{T}_j:= \mathcal{I}_k \cap \{j,\dots, k\}$, \ie the set of all ordered sequences from $\{i_{j}<\dots<i_{k-1}\}$ such that last index of each such sequence is $i_{k-1}$ and it contains at least one of two consecutive indices. This suffices as we can combine~\cref{eq: inductive} (for $j=2$) with~\cref{eq: ep} for $j=1$, to get the main claim:
	\begin{align*}
	\ep_1 ~\leq~	\lambda^{\Delta_1}\cdot \zeta_{2} ~\leq~ \lambda^{\Delta_1}\sum_{T\in \mathcal{T}_2} \lambda^{\sum_{t\in T}\Delta_t} ~=~ \sum_{I\in \mathcal{I}(k)} \lambda^{\sum_{i\in I}\Delta_i} .
	\end{align*}

	We prove~\cref{eq: inductive} by induction starting from $k-1$ which can be proved by combining~\cref{eq:base_ep} and~\cref{eq: zeta} (for $j = k-1)$,
\[	\zeta_{k-1} ~\leq~ \ep_k + \lambda^{\Delta_k}\zeta_{k} ~\leq~  \lambda^{\Delta_k}.
	\]
	Now, assume the claim holds for any $k-1\geq j>2$. Then, we have:
	\begin{align*}
		\zeta_{{j-1}}&~\leq~ \lambda^{\Delta_{{j-1}}} \cdot  \sum_{T\in \mathcal{T}_j} \lambda^{\sum_{t\in T}\Delta_t} + 
		\lambda^{\Delta_{{j}}} \cdot  \sum_{T\in \mathcal{T}_{j+1}} \lambda^{\sum_{t\in T}\Delta_t} 
		%\\&~=~ \sum_{T\in \mathcal{T}(i_{j})} \lambda^{\sum_{t\in (i_{j-1}, T)}\Delta_t} + 
		%\sum_{T\in \mathcal{T}(i_{j+1})} \lambda^{\sum_{t\in (i_{j}, T)}\Delta_t}
		\\[0.7em]&~=~ 
		\sum_{T\in \mathcal{T}_{j-1}~:~i_{j-1}\in T} \lambda^{\sum_{t\in T}\Delta_t} + 
		\sum_{T\in \mathcal{T}_{j-1}~:~i_{j-1}\notin T,~i_{j}\in T} \lambda^{\sum_{t\in T}\Delta_t} 
		\\[0.7em]&~\leq~  \sum_{T\in \mathcal{T}_{j-1}} \lambda^{\sum_{t\in T}\Delta_t} \qquad .
	\end{align*}
	The last line follows because the sets:
	\[ \{T\in \mathcal{T}_{j-1}~:~i_{j-1}\in T\}, \; ~\{T\in \mathcal{T}_{j-1}~:~i_{j-1}\notin T,~i_{j}\in T\},\]
	 are non-intersecting and subsets of $\mathcal{T}_{j-1}$.   
	Thus ~\cref{eq: inductive} holds, which concludes the proof.
\end{proof}

%\tnote{The proof can probably handle the biased case as well?}

\begin{corollary}[Operator version of~\cite{CPT21}]\label{cor:op_irrep} Let $X$ be a $\lambda$ expander and $\phi:V(X)\to G$ be an unbiased labeling. Let $\inds = \{i_i, \cdots, i_k\} \subseteq [n]$. 
	Then for any set of non-trivial irreps $\{\rho_1,\dots , \rho_k\}$ of $G$, 	  
	\[ \norm[\big]{\E_{\vec{x}\in \rwalk{n}} [\rho_1(\varphi(x_{i_1}))\cdots \otimes \rho_k(\varphi(x_{i_k})) ]}_\op  	~\leq~ \aprxone.\] 
\end{corollary}

\begin{proof}
	 We only need to check that a non-trivial irrep satisfies the conditions of~\cref{theo:tensor_fool}. The max operator norm is $1$ as representations map to unitary matrices. The mean zero condition holds from the fact we work with unbiased labelings and from~\cref{lem:mean_zero}. \end{proof}

\subsection{Walks on \enquote{structured} Cayley graphs}

In this section, we will specialize our results and work with a class of Cayley-like graphs. These graphs generalize a very useful property of Cayley graphs over Abelian groups, namely that, its eigenvectors are characters. This knowledge of the graph eigenvectors will enable us to sharpen our computation of the random walk expectation.

\begin{definition}[Pseudo Cayley graph]\label{def:pseudo}
	A graph $X$ is \textit{pseudo-Cayley} with respect to $G$ if there is an unbiased labeling $\varphi:X\to G$ such that for every Fourier basis element $\rho_{ij}$, the function, $\rho_{ij}\circ \phi$,  is an eigenvector of $A_X$ with eigenvalue $\lambda_\rho$. 
%	(ii) There is a homomorphism from $X \to \Cay(G,S)$ for some $S\subseteq G$. In other words, there is a map $\phi: V(X) \to G$ which is unbiased and $(u, v) \in E(X) \iff \inv{\phi(u)}\phi(v) \in S$. 
\end{definition}

 When working with such graphs, we will implicitly use such a labeling and thus write $\chi(v)$ as a shorthand for $\chi\circ\varphi(v)$. We now give two examples of pseudo-Cayley graphs.

\paragraph{Examples of pseudo Cayley graphs}  
\begin{example}[Quasi-Abelian Cayley graphs]\label{example:pseudo1}
Let $X = \Cay(G,S)$ be a graph where $G$ is any finite group and $S$ is a union of conjugacy classes, \ie if $s \in S$ then $gs\inv{g} \in S$ for every $g \in G$. Then,	
\end{example}

\begin{lemma}[{\cite[Thm. 1.1]{RKHS02}}]
Let $X = \Cay(G,S)$ be a graph where $G$ is any finite group and $S$ is a union of conjugacy classes. Then, every orthogonal basis vector, $\rho_{i,j}$ is an eigenvector of $X$ with an eigenvalue $\lambda_\rho := \frac{1}{d_\rho} \sum_s \chi_\rho(s)$. Therefore, $X$ is pseudo Cayley graph (with an identity labeling). 
\end{lemma} 

The above example can be generalized by picking a Cayley graph on a larger group $H$ and labeling it via a group homomorphism. We will need this for our lower bound, so we state it precisely. 

\begin{example}\label{example:pseudo2}
Let $H, G$ be groups such that there exists a surjective homomorphism $\phi: H\to G$. Then, the complete graph on $H$ (without self-loops), \ie $X = \Cay(H, H\setminus\{1\})$ is $G$-generalized Cayley labeled	with $\phi$ as the labeling. In particular, one may take $H = G^r$ for any $r \geq 1$. Moreover, it inherits the eigenvalues of $X$.
\end{example}

\subsubsection{Decay for walks on Pseudo Cayley graphs}
We now prove a finer bound for the decay obtained by performing walks on such an $X$. The labeling function is just the identity map on $G$. which is an unbiased labeling.  This improves~\cref{theo:main_tensor}  by providing an explicit description of the mean $\cE_X(\vrho)$ and not just a norm bound on it. This will be used to give better bounds for conjugacy-invariant function \ie \emph{class functions} in~\cref{sec:class}. We start with a key fact from representation theory.  

\begin{fact}[Decomposition of tensor representations]
	\label{fact:sumChar}
	Let, $\alpha,\beta\in \widehat{G}$ be two irreps of 
	a finite group $G$. There exists a change of basis transformation $\mathsf{N}_{\alpha,\beta}$, and non-negative integer coefficients $\{ c^{\alpha, \beta}_\gamma \mid \gamma \in\widehat{G}\}$ such that for any $g\in G$: 
	\begin{align*}
	\mathsf{N}_{\alpha,\beta}  \parens[\big]{ \alpha(g)\otimes \beta(g)} \mathsf{N}_{\alpha,\beta}^* ~&=~ \oplus_{\gamma \in \widehat{G}}  \gamma^{\oplus c^{\alpha,\beta}_\gamma}, \\
	 \chi_{\alpha}(g)\cdot \chi_{\beta}(g) ~&=~ \sum_{\gamma} c^{\alpha,\beta}_\gamma \cdot  \chi_{\gamma}(g), \\
	 c^{\alpha,\beta}_{\triv} ~&=~ \indicator{\alpha = \beta^*}.
	\end{align*}
	These coefficients are called Clebsch-Gordan coefficients for $G$. 
%	\begin{enumerate}[label=(\alph*)]
%		\item If $\oplus_{\gamma \in \widehat{G}} c^{\alpha,\beta}_\gamma \gamma$ is  the decomposition of $\alpha\otimes \beta$ as a $G$-representation.  Then, for all $g\in G$ it holds that: \[ \chi_{\alpha}(g)\cdot \chi_{\beta}(g) ~=~ \sum_{\gamma} c^{\alpha,\beta}_\gamma \cdot  \chi_{\gamma}(g). \]
		%\item  If $\oplus \delta_{k}$ is  the decomposition of $\beta$ into irreps. as a $G$-representation. Then it holds that,
		%\[   \E_{g}\brackets{\chi_{\alpha}(g)\cdot \chi_{\beta}(g)}= \sum_{k}    \]		
\end{fact}

The proof is an inductive unfolding of the expression by applying~\cref{fact:sumChar} and then using that the characters are eigenvectors.

\begin{theorem}[Precise Computation of Expectation]\label{theo:quasi}
Let $X$ be a {pseudo-Cayley} graph with respect to $G$, with eigenvalues $\{\lambda_{\alpha} \mid \alpha \in \irrep{G}\}$. Let $\mathcal{S} = \{i_1,\dots, i_k \}$ be any ordered subset of $[n]$, and $\{\rho_1,\dots, \rho_k\}$ be non-trivial irreps of $G$ and $\{\chi_i\}$ their associated characters. Then for any $k \geq 2$, 
%\Ex{\vec{x}\in X^n}{\chi_{\rho_1}(x_1)\cdots \chi_{\rho_n}(x_n)}
\[
\Ex{\vec{x} \sim \rwalk{n}}{\chi_1(x_{i_1})\cdots \chi_k(x_{i_k}))} 
 ~=~ \sum_{\gamma_1,\ldots, \gamma_{k-2} \in \irrep{G}}  \prod_{i=1}^{k-1} \parens[\Big]{c^{\rho_i, \gamma_i}_{\gamma_{i-1}}\lambda_{\gamma_i}^{\Delta_{i}(\scriptscriptstyle\inds)}} .
\]
where $\gamma_0 = \triv, \gamma_{k-1}= \rho_{k}$ are fixed in the summation. 
\end{theorem}
\begin{proof}
The proof will be by induction on $n$. Since the graph is quasi-abelian all the characters are eigenvectors of the random walk matrix, $A_X$. Thus for any fixed $x_i$, 
\begin{align}\label{eq:one_step}
	\Ex{x_{i+1}\sim^{\Delta_i}x_i}{\chi_\gamma(x_{i+1})} ~=~  { \chi_\gamma\parens[\Big]{A_X^{\Delta_i}\cdot x_{i}}} ~=~ \chi_{\gamma}(x_i)\cdot \lambda_\gamma^{\Delta_i}.
	\end{align}

We do the base case when $n = 2$ which corresponds directly to~\cref{fact:sumChar}. 
\begin{align}
\Ex{(x_1,x_2)}{\chi_{\rho_1}(x_1)\cdot \chi_{\rho_2}(x_2)} ~&=~  	
\Ex{x_1}{ \chi_{\rho_1}(x_1)\cdot  \Ex{x_2\sim^{\Delta_1} x_1}{\chi_{\rho_2}(x_2)} }  \nonumber\\
~&=~   \lambda^{\Delta_1}_{\rho_2} \cdot \Ex{x_1}{ \chi_{\rho_1}(x_1)\cdot  \chi_{\rho_2}(x_1) }	&& (\text{Using \cref{eq:one_step}}) \nonumber\\
~&=~   \lambda^{\Delta_1}_{\rho_2} \cdot \Ex{x_1}{ \sum_{\gamma} c^{\rho_1,\rho_2}_{\gamma}\chi_{\gamma}(x_1) }	&& (\text{Using \cref{fact:sumChar}})\label{eq:base_2}\\
~&=~    c^{\rho_1,\rho_2}_{\triv} \cdot \lambda^{\Delta_1}_{\rho_2}. 	 \nonumber\end{align}
The last step above uses the fact that $x_1$ is sampled uniformly from $X$ and thus for any non-trivial irrep, $\Ex{x_1}{\chi_\gamma(x_1)} = 0$.
\paragraph{ Inductive Step} For the inductive step denote $\theta:= \chi_{1}(x_1)\cdots \chi_{n-2}(x_{n-2})$. By our inductive hypothesis, for any irreducible representation $\gamma$,
\begin{equation}
\label{eq:induct2}	\Ex{\vec{x}\in \rwalk{n}}{\theta(x_1,\cdots, x_{n-2})\cdot \chi_{\gamma}(x_{n-1})} ~=~ \sum_{\gamma_1\cdots, \gamma_{n-3}} \prod_{i=1}^{n-2} \parens[\Big]{c^{\rho_i, \gamma_i}_{\gamma_{i-1}}\lambda_{\gamma_i}^{\Delta_{i}}}
\end{equation}
Here, $\gamma_{n-2} = \gamma$ by the inductive claim.  We now compute the inductive step,
\begin{align*}
	\cE_X(\chi_{\vrho}) ~&=~ \Ex{(x_1,\cdots, x_{n-2}) \sim \rwalk{n-2}}{\theta(x_1, \cdots, x_{n-2})\cdot \Ex{x_{n-1}, x_n}{ \chi_{\rho_{n-1}}(x_{n-1})\chi_{\rho_n}(x_n)} }\\
~&=~	\Ex{(x_1,\cdots, x_{n-2}) \sim \rwalk{n-2}}{\theta(x_1, \cdots, x_{n-2})\cdot \lambda^{\Delta_{n-1}}_{\rho_{n}} \cdot \E_{x_{n-1}}\brackets[\bigg]{ \sum_{\gamma_{n-2} \in \irrep{G}} c^{\rho_{n-1},\rho_n}_{\gamma_{n-2}}\chi_{\gamma_{n-2}}(x_{n-1}) } } \\
~&=~	\sum_{\gamma_{n-2}} c^{\rho_{n-1},\rho_n}_{\gamma_{n-2}} \lambda^{\Delta_{n-1}}_{\rho_{n}} \Ex{(x_1,\cdots, x_{n-1}) \sim X^{n-1}}{\theta(x_1, \cdots, x_{n-2})\cdot \chi_{\gamma_{n-2}}(x_{n-1}) } \\
~&=~	\sum_{\gamma_{n-2}} c^{\rho_{n-1},\rho_n}_{\gamma} \lambda^{\Delta_{n-1}}_{\rho_{n}} \sum_{\gamma_1\cdots, \gamma_{n-3}} \prod_{i=1}^{n-2} \parens[\Big]{c^{\rho_i, \gamma_i}_{\gamma_{i-1}}\lambda_{\gamma_i}^{\Delta_{i}}} && \hspace{-5em}\text{(Using \cref{eq:induct2})} \\
~&=~	 \sum_{\gamma_1\cdots, \gamma_{n-2}} \prod_{i=1}^{n-1} \parens[\Big]{c^{\rho_i, \gamma_i}_{\gamma_{i-1}}\lambda_{\gamma_i}^{\Delta_{i}}} .
\end{align*}
Note the from the expression $\gamma_{n-1} = \rho_{n-1}$. Thus, we obtain the RHS.
\end{proof}

We now derive two consequences from the above theorem that we will use later. We start with a simple one that we have already computed as the base case in our above proof. We write out separately as it we will utilize this case later. Moreover, it is conceptually important because it captures the operation of projecting to the space of $G$-invariants. 
\begin{corollary}\label{claim:level2_tensor}
Let $X$ be any pseudo Cayley graph.  and let $\vrho$ be such that $|\vec{\rho}|=2$ where $\rho_i=\alpha$ and $\rho_j=\beta$ for $\alpha,\beta\in \widehat{G}$  and $1\leq i<j\leq n$. Then,
\begin{align*}
\cE_X(\vrho) ~&=~0 \qquad \text{if}\;\; \alpha \neq \beta^*, \\
\cE_X(\vrho) ~&=~\lambda_\alpha(X)^{j-i}\cdot \Ex{g\sim G}{\alpha\otimes \alpha^*(g)} ~=:~ \lambda_\alpha^{j-i}\cdot  \sfM_\alpha. 
\end{align*}
Here, $\sfM_\alpha$ is a $d_\alpha^2\times d_\alpha^2$ matrix with  $\tr(\sfM_{\alpha})= 1$.  
\end{corollary}
\begin{proof}
This is the same computation as in the proof for the base case of $k=2$ but now utilizing the matrix decomposition of $\alpha\otimes \beta$ from~\cref{fact:sumChar}.
%First note that, by direct sum decomposition of $\alpha \otimes \alpha$, we have $(\alpha \otimes \alpha)(g)=\sfU\cdot\parens{Triv~\oplus_k \gamma_i}(g)\cdot \sfU^{-1}$ for some matrix $\sfU$ and set of irreducibles $\{\gamma_k\in \hat{G}\}_{k}$
%\begin{align*} 
%	\E_{\vec g \sim X^n }[\vrho(\vec g)]&=\E_{g\sim H, s_1,\dots,s_{j-i}\sim S}[\rho_i(g)\otimes \rho_j(gs_1s_2\dots s_{j-i})]
%	\\&= \lambda^{j-i} \E_{g\sim H}[\alpha(g)\otimes \beta(g)]
%	\\&=\lambda^{j-i}\cdot \sfM \qedhere
%\end{align*}
\end{proof}
	
Notice in the statement of~\cref{theo:quasi} that if we have terms with many of the $\gamma_i$ being trivial, then this expectation can be large, as $\lambda_\triv = 1$. To see this, assume the extreme case when every  $\lambda_\gamma=1$. Then, the term is just an inductive way to count the multiplicity of the trivial rep in the tensor-rep, $\rho_1 \otimes \cdots \rho_n$.	To give a better bound, we make the following important observation that as no two consecutive $\gamma_j ,\gamma_{j+1}$ can be trivial. We recall the definition of $\mathcal{I}_k$, 
\[
 \mathcal{I}_k~=~ \braces[\big]{\{1, k-1\} \subseteq I  \subseteq [k-1] \mid  \;\forall\; 1 < j < k-2,\; \{j,j+1\} \cap I \neq \emptyset}
\]

\begin{observation} Let $\rho$ be any non-trivial irrep of $G$. Then, $c^{\rho, \triv}_{\triv} = 0$. Let $\{\gamma_0, \gamma_1,\cdots, \gamma_{k-1}\}$ be a sequence of irreps such that $\gamma_0=\triv$ and $\gamma_{k-1}\neq \triv$. Define, $T_\gamma:= \{i \mid \gamma_i\neq \triv\}$.   Then,
\[
\prod_{i=1}^{k-1} {c^{\rho_i, \gamma_i}_{\gamma_{i-1}}} = 0 \quad \text{if } T_\gamma \not\in \mathcal{I}_k \quad .
\]
\end{observation}
\begin{proof}
By definition, $c^{\rho, \triv}_{\triv}$ is the multiplicty of the trivial representation in $\rho\otimes \triv$ which is zero as $\rho$ is a non-trivial irrep. Now, if $T_\gamma\not \in \mathcal{I}_k$, either $1\not \in T_\gamma$ or there exists $j$ such that $\{j-1,j\} \not \in T_\gamma$. This is because $k-1\in T_\gamma$ by definition. In the first case, $\gamma_0 =\gamma_1 = \triv$. In the second, we have $\gamma_j, \gamma_{j-1} = \triv$. So we have  that either the first term or the $j^{\textrm{th}}$-term in the product is zero. \end{proof}

This shows that the term in~\cref{theo:quasi} only sums over a subset of all possible sequences of irreps. To formalize this we make the following definition

%\begin{definition}[Valid Indices]   We define the following collection of subsets of $[k-1].$ 
%\[ \Lambda_k~=~ \braces[\big]{\{ k-1\} \subseteq I  \subseteq [k-1] \mid  \;\forall\; 1 < j < k-2,\; \{j,j+1\} \cap I \neq \emptyset}
%\] We say that a sequence of irreps $\{\gamma_1,\cdots, \gamma_{k-1}\}$ is \textit{valid} if its non-trivial indices, \ie $T_\gamma:= \{i \mid \gamma_i\neq \triv\} \in \Lambda_k$ . 
%\end{definition}

%\tr\parens{\cE_X(\vrho)} ~=~ \cE_X(\chi_{\vrho}) ~=~ \sum_{\gamma_1,\ldots, \gamma_{k-2} \in \irrep{G}}  \prod_{i=1}^{k-1} \parens[\Big]{c^{\rho_i, \gamma_i}_{\gamma_{i-1}}\lambda_{\gamma_i}^{\Delta_{i}(\scriptscriptstyle\inds)}}

\begin{corollary}\label{cor:trace}  Let $\{\rho_1, \dots, \rho_k\}$ be a set of $k$ non-trivial irreps of $G$, and $\{\chi_1,\dots \chi_k\}$ be the corresponding characters. Let $X$ be any pseudo Cayley graph on $G$. Then for any subset $\inds$ of size $k$,
\[
\abs{ \cE_{X,\inds}(\chi_{\vrho})} ~~  ~\leq~ 
 \angles[\big]{\chi_\triv, \,\chi_1\cdots \chi_k}\cdot \max_{T\in \mathcal{I}_k} \,\lambda^{\sum_{i \in T}\Delta_i(\inds)} \;\;.
\]
%	\[\tr\parens[\big]{\cE_X(\vrho)} \leq \aprxfive \cdot \Ex{g\sim G}{\chi_{\tilde{\vrho}} (g)} = \angles[\big]{\chi_\triv, \,\chi_{\tilde{\vrho}}},~~\text{ for } \eta \leq  (??) . \]
%	Note that the expression $\Ex{g\sim G}{\chi_\vrho (g)}$ has a nice interpretation as it counts the multiplicity of trivial rep in $\tilde{\vrho}$.
\end{corollary}

\begin{proof}
Recall from~\cref{theo:quasi} that, 
\begin{align*}
\cE_{X,\inds}(\chi_{\vrho}) ~&=~ \sum_{\gamma_1,\ldots, \gamma_{k-2} \in \irrep{G}}  \prod_{i=1}^{k-1} \parens[\Big]{c^{\rho_i, \gamma_i}_{\gamma_{i-1}}\lambda_{\gamma_i}^{\Delta_{i}(\scriptscriptstyle\inds)}}\\
\cE_{X,\inds}(\chi_{\vrho}) ~&=~ \sum_{\substack{\gamma_1,\ldots, \gamma_{k-2} \in \irrep{G},\\ T_\gamma\in \mathcal{I}_k}}  \prod_{i=1}^{k-1} \parens[\Big]{c^{\rho_i, \gamma_i}_{\gamma_{i-1}}\lambda_{\gamma_i}^{\Delta_{i}(\scriptscriptstyle\inds)}}\\
\abs{\cE_{X,\inds}(\chi_{\vrho})} ~&\leq~ \sum_{\substack{\gamma_1,\ldots, \gamma_{k-2} \in \irrep{G},\\ T_\gamma\in \mathcal{I}_k}} \lambda^{\sum_{i \in T_\gamma}\Delta_i(\inds)} \prod_{i=1}^{k-1} c^{\rho_i, \gamma_i}_{\gamma_{i-1}}  && (\text{For }\gamma_i \in T_\gamma,\; \abs{\lambda_{\gamma_i}} \leq \lambda )\\
~&\leq~ \max_{T\in \mathcal{I}_k}\;\lambda^{\sum_{i \in T}\Delta_i(\inds)} \sum_{\gamma_1,\ldots, \gamma_{k-2} \in \irrep{G}}  \prod_{i=1}^{k-1} c^{\rho_i, \gamma_i}_{\gamma_{i-1}}\\
~&\leq~ \max_{T\in \mathcal{I}_k}\;\lambda^{\sum_{i \in T}\Delta_i(\inds)} \cdot  \angles[\big]{\chi_\triv, \,\chi_{\tilde{\vrho}}}  .
\end{align*}

The last term just inductively counts the multiplicity of the trivial rep in the tensor representation $\rho_1\otimes\cdots \otimes \rho_k$ which is equal to $ \angles[\big]{\chi_\triv, \,\chi_1\cdots \chi_k}$ from~\cref{fact:sumChar}.
\end{proof}
	
	 %\begin{claim}
%For any $\vrho = \rho_1\otimes\cdots \otimes \rho_t$, such that $|\vrho| = k$. Let $\inds=\{i_1<\dots<i_k \}$ be the non-trivial coordinates of $\vrho$. 
	% and let $i_1$ be the first non-trivial co-ordinate in $\vrho$. 
%	Then, for quasi-abelian Cayley graphs, for any $g\in G$ it holds that,
%	\[\E_{\vec s\sim S^{n-1}} \brackets{\otimes_{i}\rho_i(gs_1\dots s_{i-1})}= \oplus_{j} c_j \beta_{j} (g),
%	\]
	%\[\E_{\vec s\sim S^{n-1}} \brackets{\otimes_{i}\rho_i(gs_1\dots s_{i-1})}= \oplus_{j} c_i \E_{\vec s\sim S^{i_1-1}} \brackets{ \beta_{i} (gs_1\dots s_{i_1-1})},
	%\]
%	where $\oplus_{j} \beta_{j}$ is the direct sum decomposition of $\vrho$ as a $G$-representation and each $c_j\in \Lambda(\inds)$. %$h=gs_1\dots s_{i_1}$ and 
%\end{claim}

%
%\begin{claim}\label{lem:rhoij}
%Let $\rho_1, \rho_2 \in \Irrep(G)$. Then, 
%$\Ex{g \in G}{\rho_1(g)\otimes \rho_2( g) } \neq 0 $	only if $\rho_1 = \rho_2^*$. 
%\end{claim}
%\begin{proof}
%	Let $\rho_1\otimes\rho_2 = \oplus_{\rho \in \Irrep(G)} a_\rho \rho $. Thus, one can block-diagonalize the space with $a_\rho$ copies of $\rho$. Now, since $\sum_g \rho(g) = 0$ for any non-trivial irrep $\rho$, this term is non-zero iff the trivial representation has a non-zero multiplicity. This happens iff $\rho_i = \rho_j^*$. Moreover, the multiplicity is one.
%%	 $\mathrm{triv} \subseteq \rho_1\otimes\rho_2 $ 	
%\end{proof}

%Consider $H \subseteq G$ and $\varphi: G\to H$

\section{Fooling Symmetric Functions and Word Functions}

The main goal is to study the pseudorandomness of expander walks via families of test functions. For a function, $f:G^n\to \C^{k \times k}$, we wish to analyze
\[
\cE_X(f) = \Ex{\vec{x} \sim \rwalk{n}}{f(x_1,\dots, x_n)} - \Ex{\vec{x} \sim \mathrm{Unif}_{n}}{f(x_1,\dots, x_n)} \; .
\]	

We have already analyzed this for tensor functions (\cref{theo:main_tensor}). Using Fourier transform, we will first see how studying the fooling of arbitrary functions reduces the problem to measuring the fooling of tensor product of irreducible representations.  		 
%Moreover, for abelian groups, we get a bound of $O(|G|\lambda)$ and thus improve the bound from $O(p^p\lambda)$ to $O(p\lambda)$ for $G =\Z_p$.   T
\subsection{A general reduction to fooling irreps}

\begin{claim}\label{claim:fooling}
	Let, $f:G^n \rightarrow \C$ be any function, and denote its degree-$i$ component as $f_i$. Then,
	\begin{align*}
		\cE_X(f) ~&=~ \sum_{i\geq 2}  \cE_X(f_i), \;\;\text{and}\\
		\abs[\big]{\,\cE_X(f_i)} ~&\leq~  \sum_{\vrho, |\vrho| = i}d_{\vrho} \norm[\big]{f(\hat{\rho})}_\tr \cdot\norm[\big]{\cE_X(\vrho)}_\op, \; \forall i \in [n]. 
	\end{align*}
	%	\[\cE_X(f) ~=~ \abs[\Big]{\sum_{i\geq 2} \Ex{\vec g\sim X^n}{f_i(\vec g)}} ~\leq~ \sum_{\vrho, |\vrho| \geq 2}d_{\vrho} \norm{f(\hat{\rho})}_1\cdot\norm{\E_{\vec{g}\in X^t} [\rho(\vec{g})]-\E_{\vec{g}\in \mu^t} [\rho(\vec{g})]}_\op . \]
\end{claim}
\begin{proof} 
	By definition $f = \sum_{i=0}^n f_i$, and by linearity $\cE_X(f) = \sum_{i = 0}^n  \cE_X(f_i)$. We will first show that $\cE_X(f_0) = \cE_X(f_1) = 0$ to prove the first claim.
	
	The level $0$ function, $f_0$ is a constant and hence its mean is the same under any distribution. Thus, $\cE_X(f_0) = \nu(f_0) - \mu(f_0) = 0$. Applying the Fourier transform to level $i$,
	\begin{align*}
		f_i ~&=~  \sum_{\vrho\in \mathrm{Irrep}(G^n), \abs{\vrho} = i} d_\vrho\ip{\hat{f}(\rho)}{ \rho(\vec{g})},\\
		\mu(f_i) ~&=~  \sum_{\abs{\vrho} = i} d_\vrho\ip{\hat{f}(\vrho)}{ \mu(\vrho)} \; .
	\end{align*}
	For level $1$, $\vrho(\vec{g}) = \rho_k(g_k)$ for some $k$, and by ~\cref{fact:single}, $\nu_n(\vrho(\vec{g})) = \gamma_k(g_k) = \nu_1(g_k) =  \mu(g_k)$. Thus, $ \nu_n(\vrho) =  \mu^n(\vrho)$, and $\cE_X(f_1) = 0$. Now, to obtain the second inequality, we use the above equation for our two distributions (the random walk, $\nu$, and the product $\mu^n$): 
	\begin{align}
		\cE_X(f_i) ~&=~ {\sum_{\abs{\vrho} = i} d_\vrho\ip{\hat{f}(\vrho)}{ \nu_n(\vrho) - \mu^n(\vrho)   }_{\HS}}\\
\label{eq:level_i_fool}		~&=~ {\sum_{\abs{\vrho} = i} d_\vrho\ip{\hat{f}(\vrho)}{ \cE_X(\vrho)   }_{\HS}}\\
		\abs{\cE_X(f_i)} ~&\leq~ {\sum_{\abs{\vrho} = i} d_\vrho\norm[\big]{f(\hat{\rho})}_\tr \cdot\norm[\big]{\cE_X(\vrho)}_\op}.
		%~&=~ {\sum_{\vrho\in \mathrm{Irrep}(G), \vrho\neq 1} d_\vrho\ip{\hat{f}(\vrho)}{ \E_{\vec{g}\in X^n} [\vrho(\vec{g})]}_{\HS} }
		%\\
		%~&=~ {\sum_{\vrho\in \mathrm{Irrep}(G^n), \vrho: |\vrho|\geq 2} d_\vrho\ip{\hat{f}(\rho)}{ \E_{\vec{g}\in X^n} [\vrho(\vec{g})]}_{\HS} }
		%\\
		%~&\geq~ \abs[\Big]{\sum_{\vrho\in \mathrm{Irrep}(G^n), \vrho: |\vrho|= 2} d_\vrho\ip{\hat{f}(\rho)}{ \E_{\vec{g}\in X^n} [\vrho(\vec{g})]}_{\HS} }~-~\abs[\Big]{\sum_{\vrho\in \mathrm{Irrep}(G^n), \vrho: |\vrho|\geq 3} d_\vrho\ip{\hat{f}(\rho)}{ \E_{\vec{g}\in X^n} [\vrho(\vec{g})]}_{\HS} }
	\end{align}	
	The last inequality follows by combining Von Neumann's trace inequality and H\"{o}lder's inequality. 
\end{proof}

%\snote{I thought the last inequality is just Holder's.}

%\begin{remark}
%The inequality can be strengthened for 	
%\end{remark}

\subsection{Fooling symmetric functions}
Let $f:G^n \to \C$ be a function that is invariant under any permutation of the input tuple. Such a function only depends on the counts of each group element in the tuple and, therefore, can be viewed as a symmetric function on $\Z_{|G|+1}^n$. Appealing to the results of Golowich--Vadhan~\cite{GV22}, one gets a decay of $O(|G|^{O(|G|)} \lambda)$. We obtain an exponentially better bound of $O(|G|\lambda)$ by utilizing a Fourier basis for $G$.

\paragraph{Preparatory lemmas} In the Boolean case, the Fourier coefficient of a symmetric function $f$, is unchanged under permutation of the non-trivial coordinates, \ie $\hat{f}(\chi_T) = \hat{f}(\chi_{T'})$ for any subsets $T,T'$ of size $k$. Unsurprisingly, this  extends to the case of general groups.

\begin{observation}[Fourier Coefficient under permutation]\label{obs:symmetry}
Let $\rho_1,\cdots \rho_k$ be any $k$ non-trivial irreps and let $T = \{t_1,\cdots , t_k\} $ be some ordered subset of $[n]$. Denote by $\vrho_T$ the irrep with $(\vrho_T)_{t_j} = \rho_j$  and trivial otherwise.  Let $\sigma$ be the permutation that maps $T \to T'$ for any other $T$ of size $k$. Then for any symmetric function $f$,
\begin{align*}
	\hat{f}(\vrho_T) ~&=~ \Ex{\vec{x}\in G^n} {f(\vec{x})\; \rho_1(x_{t_1}) \otimes \cdots \otimes\rho_k(x_{t_k})} \\
	~&=~ \Ex{\vec{x}\in G^n} {f(\inv{\sigma}\cdot \vec{x})\; \rho_1(x_{t_1}) \otimes \cdots \otimes\rho_k(x_{t_k})} \\
	~&=~ \Ex{\vec{y}\in G^n} {f(\vec{y})\; \rho_1(y_{{\sigma}t_1}) \otimes \cdots \otimes\rho_k(y_{{\sigma}t_k})} \\
	~&=~ \Ex{\vec{y}\in G^n} {f(\vec{y})\; \rho_1(y_{t_1'}) \otimes \cdots \otimes\rho_k(y_{t'_k})} ~=~ \hat{f}(\vrho_{T'}). 
\end{align*}
In particular, all norms are preserved. 
\end{observation}

We now obtain a trivial upper bound on the trace-norm of the Fourier transform, in terms of the $L^2$ norm. This is a fairly standard application of Cauchy-Schwarz.   
%This is a slight technicality that does not arise in the abelian case.   

\begin{lemma}[Trace norm to $L_2$-norm]\label{lem:trace_norm}
	For any symmetric function $f: G^n\to \C$,
	\[
	\sum_{\vrho = \rho_1\otimes\cdots \rho_k\otimes I\dots \otimes I, \rho_i \neq \triv} \, d_\vrho\norm[\big]{\hat{f}(\vrho)}_\tr ~\leq~ \frac{{|G|}^{\frac{k}{2}}}{\sqrt{n \choose k}} \cdot \norm{f_k}_2 \;. 
	\]
	%Additionaly, if $G$ is $D$-quasirandom, then, $\norm{f_k}_2 ~\leq~ D^{-k/2} \cdot \norm{f}_2 $.
\end{lemma}
\begin{proof} Let $\ell$ denote the LHS of this inequality. By Cauchy-Schwarz inequality,
\begin{align*}
	\ell^2 ~&\leq~ 	\sum_{\vrho = \rho_1\otimes\cdots \rho_k\otimes I\dots \otimes I, \rho_i \neq \triv} \, d_\vrho^2 	\sum_{\vrho = \rho_1\otimes\cdots \rho_k\otimes I\dots \otimes I, \rho_i \neq \triv} \,  \norm[\big]{\hat{f}(\vrho)}_\tr^2\\
	~&\leq~ 	\parens[\Big]{\sum_{\rho \neq \triv} \, d_\rho^2}^k 	\sum_{\vrho = \rho_1\otimes\cdots \rho_k\otimes I\dots \otimes I, \rho_i \neq \triv} \,  \norm[\big]{\hat{f}(\vrho)}_\tr^2\\
~&=~ 	\parens[\Big]{|G|-1}^k 	\sum_{\vrho = \rho_1\otimes\cdots \rho_k\otimes I\dots \otimes I, \rho_i \neq \triv} \,  \norm[\big]{\hat{f}(\vrho)}_\tr^2\\
~&=~ 	{|G|}^k \cdot {n \choose k}^{\scriptscriptstyle-1} \cdot 
	\sum_{\vrho\in \Irrep(G^n) :~|\vrho|=k} \,  \norm[\big]{\hat{f}(\vrho)}_\tr^2	\quad .
	\end{align*}

The last equality above utilizes the symmetry of $f$ and applies~\cref{obs:symmetry}. To compute the remaining term, we  use H\"{o}lder's inequality, which implies that  
\[\norm[\big]{\hat{f}(\vrho)}_\tr^2 ~\leq~ {d_\rho} \norm[\big]{\hat{f}(\vrho)}_\HS^2.\]
Plugging this in and using the definition of $f_k$ finishes the proof:
\begin{align*}
	\ell^2
	~&\leq~ 	{|G|}^k \cdot {n \choose k}^{-1} \cdot 
	\sum_{\vrho\in \Irrep(G^n) :~|\vrho|=k} \,  \norm[\big]{\hat{f}(\vrho)}_\tr^2\\
		~&\leq~ 	{|G|}^k \cdot {n \choose k}^{\scriptscriptstyle-1} \cdot 
	\sum_{\vrho\in \Irrep(G^n) :~|\vrho|=k} \,  d_{\vrho}\norm[\big]{\hat{f}(\vrho)}_\HS^2\\
	~&=~ 	{|G|}^k \cdot {n \choose k}^{\scriptscriptstyle-1} \cdot 
	\norm{f_k}_2^2. \qedhere
\end{align*}
%
%
%\[
%\norm{f_k}^2 ~=~ \sum_{\vrho\in \Irrep(G^n) :~|\vrho|=k} \, d_\vrho \norm[\big]{\hat{f}(\vrho)}_{\HS}^2 ~=~ {n \choose k} \sum_{\vrho = \rho_1\otimes\cdots \rho_k\otimes I\dots \otimes I, \rho_i \neq \triv} \, d_\vrho \norm[\big]{\hat{f}(\vrho)}_{\HS}^2
%\]
%The claim follows by the observation that for any symmetric function:
%\begin{align*}
% \binom{n}{k} \cdot d_\vrho \norm[\big]{\hat{f}( \vrho)}_\tr=\sum_{\sigma} \, d_\vrho\norm[\big]{\hat{f}(\sigma \vrho)}_\tr
%\end{align*} 
\end{proof}

We now recall the key combinatorial bound from~\cite{CPT21}

\begin{lemma}[{{\cite[Lemma 4.4]{CPT21}}}]\label{lem:beta_bound} For any $2 \leq k \leq n$, and $\lambda <1/2$ we have
	\[\beta (k) =  \sum_{|S| = k } \lambda^{\Delta(S)/2} ~\leq~ 2^k {n-1 \choose \lfloor k/2\rfloor}\parens[\Big]{\frac{\lambda}{1-\lambda}}^{k/2} ~\leq~ {{n \choose k}^{\frac{1}{2}}} (16e\lambda)^{k/2}. \]
\end{lemma}
\begin{proof}
	The first inequality is from the reference and the second follows by observing that \[
	{n-1 \choose \lfloor k/2\rfloor }^2 \leq~ {{n \choose k}} (2e)^{k},\;\; \text {and for }\; \lambda <1/2, \; \parens[\bigg]{\frac{\lambda}{1-\lambda}} ~\leq~ 2\lambda. \qedhere
	\]
\end{proof}

%\tnote{Do not write this corollary but use it}
%\begin{corollary}
%	For any $\vrho = \rho_1\otimes\cdots \otimes \rho_n$, such that $|\vrho| = k$,
%	 \[\sum_{\theta} \norm{\E_{\vec{g}\in X^t} [\theta\rho(\vec{g})]}_{\op} \leq \beta(k) \leq  \sqrt{{n \choose k}} (16e\lambda)^{k/2} .\]
%\end{corollary}
%\begin{proof}
%	$\beta (k) =  \sum_{|S| = k } \lambda^{\Delta(S)/2} = $ and the right hand side bound is from~\cite{CPT21}. 
%\end{proof}

%\begin{lemma}
%Assume that the group is $D$-quasirandom and let $v(f,k)$ be a vector indexed by $\{\hat{f}(\rho_1\otimes\rho_2\cdots \rho_k\otimes I \cdots \otimes I) \mid \rho_i \in \Irrep(G)\setminus \{1\} \}$. Then, $\norm{v(f,k)}_1 \leq \frac{(|G|-1)^{k/2}}{\sqrt{ {n \choose k} D^k}} $ 
%\end{lemma}
%\begin{proof}
%\begin{align*}
%	1 =\norm{f}_2^2 ~&=~ \sum_{\vrho} d_{\vrho} \norm{\hat{f}(\vrho)}_2^2 \\
%	~&=~  \sum_{k}\sum_{\rho_1,\cdots, \rho_{k} \in \Irrep(G)\setminus{1}}d_{\rho_1}\cdots d_{\rho_k}  \sum_{\theta}  \norm{\hat{f}(\theta (\rho_1\otimes \cdots \otimes\rho_k \otimes I\otimes\cdots\otimes  I ))}_2^2\\
%	~&\geq ~\sum_{k}{n \choose k} D^k \norm{v(f,k)}_2^2	
%\end{align*}
%	Therefore, $\norm{v(f,k)}_2 \leq \frac{1}{\sqrt{ {n \choose k} D^k }}$. Clearly, the dimension of $v(f,k) = \sum_{\rho_1,\cdots, \rho_k} d_{\rho_1}^2\cdots d_{\rho_k}^2 = (\sum_\rho d_\rho^2 )^k = (|G|-1)^k$. Thus, $\norm{v(f,k)}_1 
%	\leq (|G|-1)^{k/2} \norm{v(f,k)}_2  $
%\end{proof}

\begin{theorem}\label{theo:main_fool}
	Let $f$ be any symmetric function over $G^n$ where $G$ is any finite group. Let $\tau = {16e\lambda |G|}$. Then, for any $k \geq 2$,
	\[ \abs{\cE_X(f_k)} ~\leq~ \tau^{k/2}\cdot \norm{f}_2. \text{ And thus,}\; \abs{\cE_X(f)} ~\leq~ 2\tau\cdot\norm{f}_2, \; \text{if}\; \tau < 1.\]
	%$\cE(f) \leq O\left( \frac{\lambda |G|}{D} \right)$	.
\end{theorem}
\begin{proof}
	We will use $\vrho_S$ to denote the representation given by $\rho_i = $.
	\begin{align*}
		\abs{\cE(f_k)} ~&\leq~ ~\sum_{\vrho \in \Irrep(G^n), \abs{\vec{\rho}} = k} d_\vrho \norm{\hat{f}(\vrho)}_\tr\cdot \norm[\big]{\cE_X (\vrho)}_\op &&(\text{Using \cref{claim:fooling}})\\[0.7em]
%		~&\leq~  \sum_{\vrho}  d_\vrho \sum_{S \in {n \choose k} } \norm{\hat{f}(\vrho_S)}_\tr \cdot \norm{\cE_X (\vrho_S)}_\op \\
		~&\leq~~  \sum_{\vrho}  d_\vrho \norm{\hat{f}(\vrho)}_\tr \sum_{S \in {n \choose k} }  \norm{\cE_X (\vrho_S)}_\op &&\text{(Using \cref{obs:symmetry}})\\
		~&\leq~ \sum_{\vrho}  d_\vrho \norm{\hat{f}(\vrho)}_\tr \sum_{S \in {n \choose k} }  \lambda^{\Delta(S)/2} &&\text{(Using~\cref{cor:op_irrep})}\\
		~&\leq~ \sum_{\vrho} d_\vrho \norm{\hat{f}(\vrho)}_\tr \cdot {{n \choose k}^{\frac{1}{2}}} (16e\lambda)^{k/2} &&\text{(Using \cref{lem:beta_bound})} \\
		~&\leq~  (16e\lambda)^{k/2} \cdot {{|G|}^{k/2}} \cdot \norm{f}_2&&\text{(Using  \cref{lem:trace_norm})}\\
		%	~&\leq~ \sum_k \beta(k) {D^k_m} \sum_{\vrho} \norm{\hat{f}(\vrho)}_1 = \sum_k \beta(k) {D^k_m} \norm{v(f,k)}_1 \\
		%	~&\leq~  \sum_k \frac{\beta(k)}{\sqrt{n\choose k}} \left(\frac{D_m^2 |G|}{D}\right)^{k/2} 
%		~&~\leq(16e\lambda)^{k/2} \cdot {{|G|}^{k/2}} \cdot \norm{f}_2&&\text{(Using \cref{lem:beta_bound})}\\
		~&~= (\tau)^{k/2} \norm{f}_2.
		%	~&~\leq~ \frac{32e\lambda |G|}{D}\cdot \norm{f}_2 ~=~ O\left(\lambda\cdot \frac{ |G|}{D} \right) \cdot  \norm{f}_2. 
	\end{align*}
	
	To get the last inequality, we use \cref{claim:fooling} and obtain that:
	\[
	\abs{\cE_X(f)} ~\leq~ \sum_{i \geq 2} \abs{\cE_X(f_i)} ~\leq~ \parens[\Big]{\sum_{k\geq 2} \tau^{k/2}} \norm{f}_2  ~\leq~  2\tau\norm{f}_2 \quad \text{ if } \tau < 1 .\qedhere
	\]
%	 $\abs{\cE_X(f)} \leq \sum_{i \geq 2} \abs{\cE_X(f_i)}$.
%	If $\tau < 1$, then $\sum_{k\geq 2} \tau^{k/2} \leq 2\tau$,  and thus the second claim follows.
\end{proof}

%\todo{Polish the text in the proof above}

%\todo{Can we add some specific examples to illustrate the below case? Remove the remark:}
%
%\begin{remark}
%	For this to work, we only need that $\norm{\hat{f}(\rho_T)} \sim \norm{\hat{f}(\rho_{T'})}$ instead of the entire fourier coefficient being the same. In the abelian case, this is quite restrictive but it might be interesting question would be to  covers more than symmetric functions. For example for $G = \mathrm{Sym}_n$ the natural representation,  	
%\end{remark}

%For abelian groups and the group $\Z_p$ in particular, $D = 1$. Thus we get $O(p\lambda)$ instead of $O(p^p\lambda)$. For quasirandom groups such as $G = \SL_2(\F_q)$, we have an improvement as $D = \Omega(|G|^{\frac{1}{3}})$ group, and thus we get a bound sublinear in $|G|$. 
\subsection{Word functions}
A \emph{word map} of a finite group $G$ is an element of the free group on $G$. 
Given any $h: G\rightarrow\C$ and a word map $w: G^n\rightarrow G$, one can consider the composed map $h(w(\cdot)): G^n\rightarrow \C$, which is commonly referred to as a \emph{word function}. Word functions are ubiquitous in mathematics and computer science literature.

%\textit{Word functions} are functions of the form $f(x_1,\cdots, x_n) = h(x_1x_2\cdots x_n)$ where $h: G\to \C$. These functions 

%In this section, we define a very natural class of functions, that we term \emph{word functions}. 

The main result of this section is to give a complete characterization of the Fourier spectrum of a certain subclass of word functions that will be termed \emph{monomial word functions}. In particular, first we will show that these have Fourier support on the highest level and thus are analogs of the PARITY function over $\Z_2^n$. Moreover, this support is also sparse. Combining this with~\cref{cor:op_irrep}, we deduce that such functions are exponentially fooled by expander walks.

%\todo{Change $S$ or assume max support?}

\begin{definition}[Monomials and Word function]
	For an ordered subset $\inds \subseteq [n]$, a \textit{word map} of \emph{degree} $k=|S|$ is a $G$-valued function $w_{\inds}: G^n\to G$, defined as $w_{\inds} = \prod_{s\in S} g_s^{e_s}$ where $e_S \in \Z$. %The \textit{degree} of a word is the size $|\inds|$.
	 A word is \textit{monomial} if the variables are non-repeating and the exponent is $\pm 1$. A function $f:G^n \to \C$ is a \textit{monomial word function} of degree $k$, if $f = h(w(g_1,\cdots, g_n))$ for a monomial word $w$ of degree $k$ and a function $h: G\rightarrow \C$. 
\end{definition}

In the second half of this section, we consider a subclass  of functions within monomial word functions that we call \emph{monotone word functions}. Essentially, these are word functions for which corrresponding word, $w$ is monotone \ie $w = x_{i_1}\cdots x_{i_k}$ for $i_1\leq \cdots i_k$. We already mentioned that 
for monomial word functions gets fooled by expander walks upto an exponentially decaying error. However, the error bound has dependence on $|G|$. For monotone word functions we remove this dependence while achieving the same decay in terms of expansion.

\subsubsection{Fourier Spectrum of Word functions}
We begin by proving a structural claim about the fourier coefficients of word functions. The claim that we prove below essentially says that a word function 
$f:G^n\rightarrow \C $ that only utilizes a subset $\inds\subseteq [n]$ of the input co-ordinates is only supported on representations $\vrho$ such that 
$\rho_i = \triv$ for $i\notin \inds$ and $\rho_i \in \{\rho,\rho^*\}$ otherwise. Note that, the fourier  structure of these word functions on more general groups closely resemble the special case of parities.

\begin{lemma}[Fourier Mass Support]\label{lem:repsame}
	Let $f: G^n\rightarrow \C$ be a word function of degree $k$ corresponding to a set $\inds$. Let $\inds^+$ (resp. $\inds^{-}$) be the subset of elements such that $e_s = 1$ (resp., $-1$).
	Then, $\hat{f}(\vrho) \neq 0$ only if 
	\begin{enumerate}
		\item For every $i \not\in \inds$, $\rho_i = \triv$.
		\item For every $i \in \inds^+$ $\rho_i = \rho$ for some $\rho \in \Irrep(G)$.
		\item  For every $i \in \inds^-$ $\rho_i = \rho^*$ for the same $\rho$ as above.
	\end{enumerate}
\end{lemma} 
\begin{proof}
 We will permute the operators and push the spaces we argue about to the end.	Claim $(1)$ is quite easy to prove and intuitive as the function does not depend on variable $i$ and thus cannot have mass on those irreps. Formally, let $i \not \in \inds$. Then, 
	\begin{align*}
		\hat{f}(\vrho) ~&=~ \Ex{\vec g_{i}}{\rho_1(g_1)\otimes \cdots \otimes \rho_n(g_n)\otimes \Ex{g_i}{\rho_i(g_i)f(\vec{g}) }  }\\
		~&=~ \Ex{\vec g_{i}}{\rho_1(g_1)\otimes \cdots \otimes \rho_n(g_n)\otimes f(\vec{g}) \;\Ex{g_i}{\rho_i(g_i)}} = 0.
	\end{align*}
	
	Any $h: G \to \C$ can be written as $h(x) = \sum_t h(t)\mathbb{I}[x = t]$. Since the Fourier transform is linear it suffices to prove it for the indicator functions. Thus, we let $f(\vec{g}) = \mathbb{I}[w(\vec{g}) = t]$, and pick any $i < j \in \inds^+$. Let $w(\vec{g}) = w_1 g_i w_2 g_j w_3 = t$. Assume that $\rho_i \neq \rho_j$. 
	\begin{align*}
		\hat{f}(\vrho) ~&=~ \Ex{\vrho_{i,j}}{\rho_1(g_1)\otimes \cdots \otimes \rho_n(g_n)\otimes \Ex{g_i, g_j}{\rho_i(g_i)\otimes \rho_j(g_j) \mathbb{I}[w(g) = t] }  }\\
		~&=~ \Ex{\vrho_{i,j}}{\rho_1(g_1)\otimes \cdots \otimes \rho_n(g_n)\otimes \Ex{g_i, g_j}{\rho_i(g_i)\otimes \rho_j( w_2^{-1}g_i^{-1}w_1^{-1}tw_3^{-1} ) }  }	\\
		~&=~ \Ex{\vrho_{i,j}}{\rho_1(g_1)\otimes \cdots \otimes \rho_n(g_n)\otimes
			\brackets[\Big]{   (I \otimes \rho_j(w_2^{-1}))   \Ex{g_i, g_j}{\rho_i(g_i)\otimes \rho_j( g_i^{-1}) } (I\otimes \rho_j(w_1^{-1}tw_3^{-1}) )} }\\
		~&=~ \Ex{\vrho_{i,j}}{\rho_1(g_1)\otimes \cdots \otimes \rho_n(g_n)\otimes
			\brackets[\Big]{ (I \otimes \rho_j(w_2^{-1}))   \Ex{g_i, g_j}{\rho_i(g_i)\otimes \rho_j^*( g_i) } (I\otimes \rho_j(w_1^{-1}tw_3^{-1}) )} } \\
		~&=~ 0 \qquad\qquad\;\; \small{\text{(Using \cref{fact:sumChar})}}.	
	\end{align*}
	
The proof for claim $(3)$ is identical but now we consider $i \in \inds^+,\, j \in \inds^{-}$. The term is then $\rho(g_i)\otimes \rho_j(g_j)$ which is non-zero unless $\rho_i = \rho_j^*$. 
\end{proof}

\subsubsection{Fooling Word Functions}

Before we state our first theorem in this section we recall two notations from~\cref{subsec: foolingtensor}.
Let $\inds  = \set{i_1 < i_2<\cdots < i_{k-1}< i_k}$ be an ordered subset of $\{1,2,\dots,n\}$. We define the following key quantities: 
\begin{itemize}
	\item $	\mathcal{I}_k ~=~ \braces[\big]{\{1, k-1\} \subseteq I  \subseteq [k-1] \mid  \;\forall\; 1 < j < k-1,\; \{j,j+1\} \cap I \neq \emptyset}$. 
	\item $\Delta_{j}(\inds) ~=~ i_{j+1}-i_{j}$. 
	%\item $\Delta(\inds) ~=~ \sum^{k-2}_j\min(\Delta_j,\Delta_{j+1})$ 
\end{itemize}
% First, given an ordered set $\widetilde{\inds} = \set{i_1 < i_2<\cdots < i_{k-1}< i_k}$ of $\{1,2,\dots,n\}$, we defined $\Delta_{j}(\widetilde{\inds} ) ~=~ i_{j+1}-i_{j}$. We also defined the set family $	\mathcal{I}_k ~=~ \braces[\big]{\{1, k-1\} \subseteq I  \subseteq [k-1] \mid  \;\forall\; 1 < j < k-1,\; \{j,j+1\} \cap I \neq \emptyset}$.  
We have the following theorem on monomial word functions.
	
%Now, we can state our main theorem that we prove in this section. 

%\begin{lemma}[Precise Fourier Coefficients] Let the word function be $f(\vec{x}) = \mathbb{I}[w(\vec{x}) = t]$. Then, $\hat{f}(\vrho) =  \Ex{g \in G}{\vrho(g,g,\cdots, g)}$
	
%\end{lemma}

\begin{theorem}[Fooling for degree $k$ word functions]\label{theo:word_func1}
	Let  $f: G^n\rightarrow \C$ be a monomial word function of degree $k$ corresponding to a set $\inds$. 
	Then for any expander $X$ with an unbiased $G$-labelling,
	\[
	 \abs{\cE_X(f)} ~\leq~ \aprxone \cdot |G|^{k/2}\cdot \norm{f}_2\;. \]  	
In particular, we have $\abs{\cE_X(f)} ~\leq~\lambda^{-1}\parens{\lambda |G|}^{k/2}\cdot \norm{f}_2 $
\end{theorem}

%\todo{change notation in the proof} 
\begin{proof}
	Let $\rho^S$ denote the representation $\rho_1\otimes\rho_2\otimes\cdots\otimes \rho_n$ where $\rho_i = \rho$ if $i \in S_I^+$, $\rho_i = \rho^*$ if $i \in S_I^-$, and  is trivial otherwise, \ie $\rho_i = 1$. From~\cref{lem:repsame}, we know that the only non-zero Fourier coefficients are for such $\rho^S$. Thus we will consider these irreps for expander walk fooling. 
	
	\begin{align*}
		\cE(f) ~&\leq~ \sum_{\psi \in \Irrep(G^n)} d_\psi \norm{\hat{f}(\psi)}_\tr\cdot \norm[\Big]{\Ex{\vec{g}\sim X^n}{\psi(\vec{g})}}_\op && \text{(Using  \cref{claim:fooling})} \\
		~&\leq~ \sum_{\rho\in \Irrep(G)}  d_{\rho^S}  \norm{\hat{f}(\rho^S)}_\tr \cdot\norm[\Big]{\Ex{\vec{g}}{\rho^S(\vec{g})}}_\op && \text{(Using~\cref{lem:repsame})} \\
		~&\leq~ \aprxone \sum_{\rho\in \Irrep(G)} 
		 d_\rho^k  \norm{\hat{f}(\rho^S)}_\tr  && \text{(Using~\cref{cor:op_irrep})}\\
	%~&\leq~ 
	%\aprxone \cdot \sqrt{\sum_{\rho\in \Irrep(G)}  d_\rho^{2k}  }\cdot \sqrt{\sum_{\rho\in \Irrep(G)}  \norm{\hat{f}(\rho^S)}^2_\tr    }
	%\\
		~&\leq~\aprxone\cdot  |G|^{k/2} \cdot \norm{f}_2. 
	\end{align*} 
The last line follows from Cauchy-Schwarz and H\"older's inequality. \qedhere
\end{proof}

We now give an alternate proof of the above result in the special case that the word is monotone \ie $w = x_{i_1}\cdots x_{i_k}$ for $i_1< \cdots < i_k$. The change is that we now the Fourier decomposition of $h: G\to \C$ \ie over $G$ rather than of $f$ directly. To analyze this, we will need the result from~\cite{JMRW22}.

\begin{theorem}[{~\cite{JMRW22}}]
Let	$X$ be any $\lambda$-expander with an unbiased labelling of $G$. Then for any non-trivial irrep $\rho$ of $G$,
\[
\norm{\cE_X\parens{\rho(x_1\cdots x_k)}}_\op ~\leq~ \lambda^{k/2}. 
\]
\end{theorem}

The result in \cite{JMRW22} works for any product function and thus if $x_i$ contains an inverse then, we can pick $f_i = \rho^*$ instead of $\rho$ as $\rho^*(x_i) = \rho(\inv{x_i})$
%While the result as stated works only when $x_i$'s do not contain inverses, the theorem is exactly the same even with arbitrary inverses.

%in the case when the word . This is merely for the convenience as the result in~\cite{JMRW22} is only stated for such. However, the proof carries out even when including inverses. 

\begin{theorem}\label{theo:word_func2}
	Let $f(\vec{x}) = h(x_1\cdots x_k)$ be a monotone word function for some $h: G\to \C$. Then for any expander $X$ with an unbiased $G$-labelling,
	\[
	\abs{\cE_X(f)} ~\leq~ \parens[\Big]{\sqrt{|G|} \cdot \norm{f}_2} \cdot (2\lambda)^{k/2}.
	\]
	In particular, for $f(\vec{x}) = \indicator{x_1\cdots x_k = t}$ for any $t \in G$, one has $\abs{\cE_X(f)} ~\leq~ (2\lambda)^{k/2}$.
\end{theorem}

\begin{proof} We assume that the word does not contain inverses and is $x_1\cdots x_k$ which is true up to renumbering the coordinates.  By the Fourier transform on $G$, \[
h(t) = \sum_{\rho \in \Irrep(G)} d_\rho \angles[\big]{\widehat{h}(\rho), \, \rho(t)}.
\]

Now we feed in $t = x_1\cdots x_n$ into the function $h$.
\begin{align*}
	f(\vec{x}) = h(x_1\cdots x_k) ~&=~ \sum_{\rho \in \Irrep(G)} d_\rho \angles[\big]{\widehat{h}(\rho), \, \rho(x_1\cdots x_k)}\\
\cE_X(f)	~&=~ \sum_{\rho \in \Irrep(G) } d_\rho \angles[\big]{\widehat{h}(\rho), \, \cE_X\parens{\rho(x_1\cdots x_k)}}\\
\abs{\cE_X(f)}	~&\leq~ \sum_{\rho \in \Irrep(G), \rho \neq \triv } d_\rho \norm{\widehat{h}(\rho)}_\tr \, \norm{\cE_X\parens{\rho(x_1\cdots x_k)}}_\op\\
~&\leq~ \lambda^{k/2} \cdot \sqrt{|G|} \cdot \norm{h}_2 = \lambda^{k/2} \cdot \sqrt{|G|} \cdot \norm{f}_2. 
\end{align*}

The last equality is a simple calculation that uses that for any fixed $x_1,\ldots ,x_i$ the word $x_1\cdots x_i\cdot g$ is uniform over $G$ if $g$ is sampled uniformly from $G$.  
\[\norm{f}_2^2 = \Ex{x_1,\cdots, x_k}{|h(x_1\cdots x_k)|^2} =  \Ex{x_1,\cdots, x_{k-1}}{\Ex{x_k}{|h(x_1\cdots x_k)|^2}} = \Ex{x_1,\cdots, x_{k-1}}{\norm{h}_2^2}.\]

When $h = \ones_{x = t}$, then $\norm{h}_2 = {|G|}^{-1/2}$ and the second claim follows. 
%	\begin{align*}
%		\cE(f) ~&\leq~ \sum_{\psi \in \Irrep(G^n)} d_\psi \norm{\hat{f}(\psi)}_1\cdot \norm[\Big]{\Ex{\vec{g}\sim X^n}{\psi(\vec{g})}}_\op\\
%		~&\leq~ \sum_{\rho\in \Irrep(G)}  d_{\rho^S}  \norm{\hat{f}(\rho^S)}_1 \cdot\norm[\Big]{\Ex{\vec{g}}{\rho^S(\vec{g})}}_\op \\
%		~&\leq~  \lambda^{\Delta(S_I)/2} \sum_{\rho\in \Irrep(G)}  d_\rho^k  \norm{\hat{f}(\rho^S)}_1   \\
%%		~&\leq~  \lambda^{\Delta(S_I)/2} \sum_{\rho\in \Irrep(G)}   \norm{\hat{f}(\rho^S)}_1   \\
%		~&\leq~ \parens{{\lambda |G|}}^{k/2} . \qedhere
%	\end{align*} 
\end{proof}

\section{Function Classes with Group Symmetry}

A general group $G$ has a much richer symmetry structure than $\Z_2$, and this opens up the possibility of studying functions, $f:G^n \to \C$, that respect this additional symmetry (beyond permutation of coordinates). 
%In this subsection, we will investigate  

\subsection{Symmetric Class Functions}\label{sec:class}

A function over $G^n$ is a \textit{class function} if it is invariant under conjugation, \ie for any $\vec{x}, \vec{g} \in G^n$, $f(g_1x_1\inv{g_1}, \ldots, g_nx_n\inv{g_n})$. In other words, the function value depends only on the input's conjugacy class. In this subsection, we will give a better bound for symmetric class functions than the one for general symmetric functions.  The improvement for class functions will come from a precise calculation of our $\cE_X(f)$ expression, without resorting to a Cauchy-Schwarz--type bound to go from $L_1$-norm to $L_2$-norm (\cref{lem:trace_norm}). To do this, we need to use the group structure.

\paragraph{Representation theory facts} We now state some basic facts about the representation theory of groups that we will need only in this subsection. These are well-known facts and proofs can be found in any introductory text. 

\begin{fact}[Class Function Fourier Coefficients]
	For any class function $f:H\to \C$, 
	\[
	\widehat{f}(\vrho) = c_\vrho I_{d_\vrho}, \qquad \norm{f}_2^2 = \sum_\vrho d_\rho \norm{\widehat{f}(\vrho)}^2_\HS = \sum_\vrho d_\rho^2 {c_\vrho}^2.
	\]
\end{fact}
%$\widehat{f}(\vrho) = c_\vrho I_{d_\vrho}$. Therefore,
%$\norm{f}_2^2 = \sum_\vrho d_\rho \norm{\widehat{f}(\vrho)}^2_\HS = \sum_\vrho d_\rho^2 {c_\vrho}^2 $.

\begin{fact}[Schur Orthogonality Relation]\label{fact:schur}
For any $g, h \in G$, we denote $g \sim h$ if they belong to the same conjugacy class,  say $C_g$. Then, we have,
\[
 \sum_{\rho \in  \Irrep G} {{\chi_{\rho}(g)\overline{\chi_{\rho}}(h)}} ~=~  \frac{|G|}{|C_g|} \cdot \indicator{g\sim h}. 
\]	
\end{fact}
%The above function, namely, $\eta_G(s) := \sum_{C\in \mathcal{C}(G)} |C|^{-s}$ has been studied in the literature and it is known that this quantity is much smaller than $1$ for the class of \emph{finite simple groups of Lie type}. In particular, consider $G = \SL_2(2^r)$, the group of $2\times 2$ matrices over $\F_{2^r}$ with determinant $1$. For this group, $\eta_G(s) := \sum_{C\in \mathcal{C}(G)} |C|^{-s}  \simeq |G|^{\frac{1-2s}{3}}$. Therefore, $\eta_{k, \SL_2(2^r)} = \Theta(|G|^{\frac{k-1}{3}})$ which is a much better bound than the trivial bound of $\Theta(|G|^{{k-1}})$ which holds for Abelian groups. Therefore, one can also interpret this as a measure of the non-commutativity of the group. Indeed for $k = 2$, one has $\eta_{2,G} = |\mathcal{C}(G)| \leq \frac{G}{D^2}$ where $D$ is the dimension of the smallest non-trivial irrep of $G$. Such groups are called \textit{quasirandom groups}.

\begin{fact}\label{fact:conj}
	Let $G$ be a $D$-quasirandom, \ie the smallest non-trivial irrep has dimension $D$. Let  $\mathcal{C}(G)$ denote the conjugacy classes of $G$. Then, 
	\[\abs{\mathcal{C}(G)} ~=~ \abs{\irrep{G}}  ~\leq~ \frac{|G|}{D^2} +1.\]
\end{fact}
\begin{proof}
	The first equality follows from the fact that characters form a basis for class functions (or in other words, the character table is square). The second follows from the following: 
	\[ \sum_{\rho \in \irrep{G}} d_\rho^2 ~=~ |G| ~\geq~ 1 + D^2 \cdot \parens{\abs{\mathcal{C}(G)} -1} . \qedhere
	\]
\end{proof}

We are now ready to assemble the above facts to bound  $\Ex{g\sim G}{\chi_{\rho_1\otimes \cdots \otimes \rho_k}}$ which counts the multiplicity of trivial rep in ${\rho_1\otimes \cdots \otimes \rho_k}$. This claim allows us to improve upon \cref{lem:trace_norm} which would be analogous to a bound of $|G|^{k}$ in the below term.

\begin{corollary}\label{char_sum}
	For any finite group $G$ denote by $ \mathcal{C}(G)$ the conjugacy classes of $G$. For any $k \geq 1$, 
	\[
	\eta_{k, G}^2 ~:=~ \sum_{\rho_1, \cdots,\rho_k \in \Irrep \setminus \triv} \parens[\Big]{\Ex{g}{\chi_{\rho_1\otimes \cdots \otimes \rho_k}}}^2 	 ~\leq~ \sum_{C \in \mathcal{C}(G)} \frac{|G|^{k-2}}{|C|^{k-2}} + 1. 
	\]  
	In particular, if $G$ is $D$-quasirandom, then $\eta_{k, G}^2 \leq 4\cdot \frac{|G|^{k-1}}{D^2}$.
\end{corollary}
\begin{proof}
The proof is a simple computation given the above fact. 
	\begin{align*}
	\sum_{\rho_1, \cdots,\rho_k \in \Irrep \setminus \triv} \parens[\Big]{\Ex{g}{\chi_\vrho}}^2 
	%	~&=~ 	\sum_{\rho_i} \parens[\Big]{\Ex{c}{n_c\chi_\vrho}}^2  \\
	%	 ~&=~ \Ex{c}{\sum_{\rho_i} \parens[\Big]{{n_c\chi_\vrho}}^2}  \\
	~&=~\sum_{\rho_i} \prod_i \parens[\Big]{\Ex{g,h}{\chi_{\rho_i}(g)\overline{\chi_{\rho_i}}(h)}}\\
	~&=~\Ex{g,h}{\sum_{\rho_i} \prod_i \parens[\Big]{{\chi_{\rho_i}(g)\overline{\chi_{\rho_i}}(h)}}}\\
	~&=~\Ex{g,h}{ \prod_i \sum_{\rho_i\neq \triv}\parens[\Big]{{\chi_{\rho_i}(g)\overline{\chi_{\rho_i}}(h)}}}.
	\end{align*}
Now we use~\cref{fact:schur} and subtract off the summand corresponding to the trivial irrep (which is just $1$).
\begin{align*}
	\sum_{\rho_1, \cdots,\rho_k \in \Irrep \setminus \triv} \parens[\Big]{\Ex{g}{\chi_\vrho}}^2 	~&=~\Ex{g,h}{ \prod_i \parens[\bigg]{\frac{|G|}{|C_g|}\indicator{g\sim h} - 1 }}\\
	~&\leq ~ \Ex{g,h}{\parens[\Big]{\frac{|G|}{|C_g|}}^k\indicator{g\sim h}+ 1}  \\
	~&\leq ~\sum_{C\in \mathcal{C}(G)} \frac{|C|^{2}}{|G|^2} \frac{|G|^{k}}{|C|^k} + 1  = \sum_{C\in \mathcal{C}(G)} \frac{|G|^{k-2}}{|C|^{k-2}} + 1 .   
\end{align*}
Now, we use~\cref{fact:conj} to conclude that 
\[ 
\eta_{k,G} ~\leq~  2 \sum_{C\in \mathcal{C}(G)} \frac{|G|^{k-2}}{|C|^{k-2}} ~\leq~ 2 |G|^{k-2}\cdot \abs{\mathcal{C}(G)} ~\leq~ 4\frac{|G|^{k-1}}{D^2}.\qedhere
\]
\end{proof}

\begin{proposition}\label{prop:quasi}
Let $G$ be a $D$-quasirandom group and $f: G^n:\to \C$ be a class function that is also symmetric. Let $X$ be a pseudo Cayley graph with expansion $\lambda$.  Then,
	\[
	\abs{\cE_X\parens[\big]{f}} ~\leq~ O\parens[\Big]{{\frac{|G|^{\frac{1}{2}}}{D}} \lambda} \cdot \norm{f}_2 .
	\]
	In particular, for every symmetric function on an Abelian group, we get a bound of $O(\sqrt{|G|}\cdot \lambda)$.
%\[	\abs{\cE_X\parens[\big]{f_k}} ~\leq~ \parens[\Big]{{16e\lambda} }^{k/2} \cdot \eta_{k, G} \cdot \norm{f_k}_2 ~\leq~ \parens{\sum_{C} \frac{|G|^{k-2}}{|C|^{k-2}} + 1}^{\frac{1}{2}} .
%	\]
%	In particular for quasirandom groups such as $\SL_2(q)$, 
\end{proposition}
 \begin{proof}
	We first prove a bound for the degree $k$ component of $f$.
\begin{align*}
	\cE_X\parens[\big]{f_k}~&=~ \sum_{\vrho\in \mathrm{Irrep}(G^n)} d_\vrho c_\vrho \cdot { \cE_X(\chi_\vrho)} \\
	~&=~\sum_{\rho_1,\cdots \rho_k \in \mathrm{Irrep}(G)\setminus \triv} \sum_{S\subseteq {n\choose k} }d_\vrho c_\vrho \cdot { \cE_X(\chi_\vrho)} \\
	\abs[\big]{\cE_X\parens[\big]{f_k}} ~&\leq~\sum_{\rho_1,\cdots \rho_k \in \mathrm{Irrep}(G)\setminus \triv} {d_\vrho} \sum_{S\subseteq {n\choose k} } \abs{c_\vrho}\cdot \abs{ \cE_X(\chi_\vrho)} \\
	~&\leq~\sum_{\rho_1,\cdots \rho_k \in \mathrm{Irrep}(G)\setminus \triv} \abs{d_\vrho c_\vrho } \cdot \parens[\Big]{\Ex{g}{\chi_\vrho}} \sum_{S\subseteq {n\choose k} } \lambda^{\Delta(S)/2} \\
	~&\leq~\beta(k) \sum_{\rho_1,\cdots \rho_k \in \mathrm{Irrep}(G)\setminus \triv} \abs{d_\vrho c_\vrho } \cdot \parens[\Big]{\Ex{g}{\chi_\vrho}} \\
	~&\leq~\beta(k) \sqrt{\sum_{\rho_1,\cdots \rho_k } \abs{d_\vrho c_\vrho }^2 }\cdot  \sqrt{\sum_{\rho_1,\cdots \rho_k } \parens[\Big]{\Ex{g}{\chi_\vrho}}^2 } \\
	~&\leq~\beta(k) \cdot \frac{\norm{f_k}_2}{\sqrt{{n \choose k}}}  \cdot \sqrt{\eta^2_{k,G}}  \\
	~&\leq~ \parens[\Big]{{16e\lambda} }^{k/2} \cdot \eta_{k, G} \cdot \norm{f_k}_2 \\
	~&\leq~ \parens[\Big]{{16e\lambda} \cdot 2 |G| }^{k/2} \cdot \frac{1}{D\sqrt{|G|}} \cdot \norm{f}_2 .
\end{align*}
Therefore, 
\[
\abs{\cE_X(f)} ~\leq~\sum_{k=2}^n \abs[\big]{\cE_X\parens[\big]{f_k}} ~\leq~ \parens[\Big]{\frac{64e\lambda |G|}{D\sqrt{|G|}}} \cdot \norm{f}_2 = O\parens[\bigg]{\frac{\lambda \sqrt{|G|}}{D}} \cdot \norm{f}_2\quad.\qedhere
\]
\end{proof}

\subsection{Diagonal action and $G$-invariant functions}
\begin{definition}[Diagonal action and Projection]
Let $h \in G$ and $f:G^n\to \C$. Define $(h\cdot f) (\vec{x}) := f(hx_1, \ldots, hx_n) = f(h\cdot\vec {x})$. The projection to the space of functions invariant under this action is $P_Gf (\vec{x}) := \Ex{h \sim G}{(h\cdot f) (\vec{x})}$.
\end{definition}

This generalizes the notion of even and odd functions over $\Z_2^n$ which are the special cases when $P_G(f) = f$ and $P_G f = 0$, respectively. We now make a simple observation that walks over Cayley graphs smooth out the function via this projection.
\begin{observation}
If $X$ is a Cayley graph, then $\cE_X(f) = \cE_X(P_G f)$.	
\end{observation}

 We will now compute the Fourier spectrum of $P_G f$ and utilize this to get a precise calculation of level-2 mass. While this can be generalized to state a more general claim, we just include the version we will need later for the lower bound.

%\begin{lemma}[Cayley graphs project]
%	Let $f:G^n\to \C$
%\end{lemma}

\begin{corollary}\label{corr:level2_proj}
	Let $f: G\times G \to \C$, and $X$ be a quasi-Ableian Cayley expander such that all non-trivial eigenvalues are $\lambda$. Then, 
	\[\cE_X(f) ~=~  \lambda \cdot \parens[\big]{P_G f(\vec{1}) - \mu(f)}.\]
\end{corollary}
\begin{proof}We start by computing the Fourier coefficient of $P_Gf$.
	\begin{align*}
	\widehat{P_Gf}(\vrho)	~&=~ \Ex{\vec{y}\in G^n}{(P_G{f})(\vec{y})\vrho(\vec {y})}\\	
	~&=~ \Ex{\vec{y}\in G^n}{ \Ex{h\in G}{{f}(\vec{h\cdot y})\,\vrho(\vec {y})}} \\
	~&=~ \Ex{h\in G}{\Ex{\vec{x}\in G^n}{{f}(\vec{x})\vrho(\inv{h}\vec {x})}} 	\\
		~&=~ \Ex{h\in G}{\vrho(\inv{h},\ldots, \inv{h} )\Ex{\vec{x}\in G^n}{{f}(\vec{x})\vrho(\vec {x})}} \\
		~&=~ \Ex{h}{\vrho(h)} \widehat{f}(\rho) .
		\end{align*}
		Observe that if $\vrho$ is trivial, then the Fourier coefficient is unchanged ie the mean of the function remains the same. If $|\vrho| = 1$, then, $\Ex{h}{\vrho(h)} = 0$. For level-2, we use~\cref{claim:level2_tensor} to say that $\Ex{h}{\alpha\otimes \alpha^*(h)}= \sfM$, and for $\beta\neq \alpha$, $\Ex{h}{\alpha\otimes \beta(h)}= 0$.
		Using the Fourier decomposition for the function $P_G f$:
		\begin{align}
		P_G f (1,1) ~&=~ \widehat{P_G f}(\triv) + 	\sum_\alpha  d_{\alpha\otimes \alpha} \ip{ \widehat{P_G f}(\alpha\otimes \alpha)}{\,\alpha\otimes \alpha^*(1,1)}\notag\\
		~&=~ \mu(f) + 	\sum_\alpha  d_{\alpha\otimes \alpha} \ip{\Ex{h}{\alpha\otimes \alpha^*(h)} \widehat{f}(\alpha\otimes \alpha)}{\alpha\otimes \alpha^*(1,1)}\notag\\
\label{eq:proj_ones}	P_G f (1,1) - \mu(f)	~&=~  \sum_\alpha d_{\alpha\otimes \alpha}     \ip{ \widehat{f}(\alpha\otimes \alpha)}{\Ex{h}{\alpha\otimes \alpha^*(h)}} && [\sfM = \sfM^*]\\
		~&=~ \sum_\alpha d_{\alpha\otimes \alpha}     \angles[\Big]{ \widehat{f}(\alpha\otimes \alpha), \, \frac{1}{\lambda_\alpha}\cE_X(\alpha\otimes \alpha^*)} && [\text{\cref{claim:level2_tensor}}]\\
\lambda\parens[\big]{P_G f (1,1) - \widehat{f}(\triv)} ~&=~  	\cE_X(f). && [\lambda_\alpha = \lambda, \forall \alpha]\notag
		\end{align}
		
	The last equation follows by the Fourier expression for $\cE_X(f)$ (\cref{eq:level_i_fool}) and by using \cref{claim:fooling} to nly consider level-2 coefficients.
%	 	as we assume all eigenvalues are equal.  For
%		
%We can now connect this with $\cE_X$ using the Cayley structure of  
% 
% 
% $\lambda_\alpha \sfM = \lambda \sfM$, as we assume all eigenvalues are equal.
\end{proof}

\begin{remark}
Another way to deduce the Fourier coefficient of $P_G f$ is to observe that this operator is a convolution, $P_G f = f* \ones_D$ where $\ones_D$ is the indicator of the diagonal subgroup \ie $\ones_D (\vec{x}) = |G|^{n-1}$ if $x_1= x_2\cdots= x_n$, and $0$ otherwise.	
\end{remark}

\section{Lower Bounds for decay of Symmetric functions}

%\paragraph{Representation theory claims Need to move it}

%\begin{proof}
%	\begin{align*}
%		\Ex{g,s}{\alpha(\lab (g))\otimes \beta(\lab(gs))} ~&=~ \Ex{g}{\alpha(\lab (g))\otimes \Ex{s}{\beta(\lab(gs))}}\\
%		~&=~ \Ex{g}{\alpha(\lab (g))\otimes \gamma(g)} \\
%	\end{align*}
%\end{proof}

\subsection{Fourier Coefficient of Threshold Function}

\paragraph{Threshold Function} Let, $A \subseteq G$  
and $t \in [n]$. We define a boolean function $\thresh$ as :
 \[\thresh(\vec x) = 1 \;\; \text{ if }\;  |\{i \mid x_i\in A \}| \geq t; \;\; 0 \;\;\text{otherwise}.
\]

%\[\ones^g_{>t}(\vec x) = 1 \;\; \text{ if }\;  \sum_i g(x_i) >t; \;\; 0 \;\;\text{otherwise}.
%\]
This function will be our candidate for the lower bound for an appropriate choice of $A, t$ to be decided later. We first compute the second-level Fourier coefficients that hold for any $A,t$. The level-$k$ coefficients can be computed in the same way, but we will not require it and hence omit the calculation.  

%where $\{C_i\}_i$ are a subset of conjugacy classes that will be chosen later. 
%We will focus on $\cE_X(\ones^A_{>t})$ where 

%\begin{lemma}
%	For the function $f = \thresh,$ we have $\Ex{{\vec x} \sim X^n}{f_1({\vec x})} = 0$ for any choice of $S$ and $t$.
%\end{lemma}
%\begin{proof}
%	
%\end{proof}
%For our ease, we use the notation:  
%
%\[\sfQ_2(f) = {\Ex{\vec g\sim X^n}{f_2(\vec g)}\,},\; \text{and} \;\; \sfQ_{\geq 3}(f):=  {\sum_{i\geq 3}\Ex{\vec g\sim X^n}{f_i(\vec g)}}.
%\]
%\\
%~&
%=:~\sfQ(f) - \sfQ'(f)
%We will now do some computations of the second-level Fourier coefficients of the threshold function. 

\begin{claim}\label{claim:level2fourier}
	Let, $\vec \rho=\rho_1\otimes \rho_2\otimes \dots \rho_n \in \widehat{G^n}$ such that $\rho_i=\alpha,~\rho_j=\alpha^*$ for some 
	$\alpha \in \hat{G}$ and $1\leq i<j\leq n$ and $\rho_k=\triv$ for any $ k\not\in \{i, j\} $. Then,
	\[\widehat{\thresh}(\vrho)= \parens[\bigg]{\frac{{a_{t-2}^{n-2} - a_{t-1}^{n-2}}}{|G|^{n-2}}} \cdot
  \widehat{\ones_A}(\alpha)\otimes \widehat{\ones_A}(\alpha^*)	\]
%  \[\widehat{\ones^A_{>t}}(\vrho)= C_{A,n,t,k}  \cdot \widehat{\ones_{A^n}}(\vrho)	\]
	where  $a_{t}^n := {n \choose t} \abs{A}^t\abs{A^c}^{n-t} $. 
%	Therefore, the level-2 mass of this function is 	$\frac{2\parens{a_{t}^{n-2} - a_{t-1}^{n-2}}}{|G|^{n}}\cdot (|T||T^c|)$.
	%$a^{n-1}_{t-1}\cdot \frac{(n-2-t)\cdot|S| t\cdot |S^c|}{(n-1)|G|^n}$
	%\cdot \E_{y\sim S}[\alpha(y)]\otimes \E_{y\sim S}[\alpha^*(y)]
\end{claim}
\begin{proof}
For $\vec x\in G^n$, we will denote the tuple obtained be removing $x_i$ and $x_j$ as $\vec x_{ij} \in G^{n-2}$. 
\begin{align*}
	|G^n|\cdot \widehat{\thresh} (\vec \rho) ~&=~  \sum_{\vec{x} \in G^n} \thresh(\vec x) \cdot \vec \rho (\vec x)   \\
	~&=~  \sum_{\vec{x} \in G^n} \thresh({\vec x}) \cdot  \alpha(x_i)\otimes\alpha^*(x_j)   \\
	~&=~ \sum_{x_i \in G, x_j \in G}  \alpha(x_i)\otimes\alpha^*(x_j)  \sum_{\vec{x_{ij}} \in G^{n-2}}   \thresh({\vec x}). 
\end{align*}
	Now observe that if $r = \abs{\{x_i, x_j\} \cap A}$ then, $\thresh{\vec x} ~=~ \threshold{t-r}{\vec x_{ij}}$. Moreover,   
	\[\sum_{\vec{x_{ij}} \in G^{n-2}}\threshold{t-r}(\vec x_{ij}) ~=~ \sum_{k = t-r}^{n-2} a^{n-2}_{k} ~=:~ a^{n-2}_{\geq t-r}.
	\]
	Hence, we have $3$ possible cases corresponding to each $r \in \{0,1,2\}$ and a total of four terms:
	\[
	|G^n|\cdot \widehat{\thresh} (\vec \rho) ~=~ \sum_{r=0}^{2}a^{n-2}_{\geq t-r}\sum_{|\{x_i, x_j \}\cap A| = r}  \parens{\alpha(x_i)\otimes\alpha^*(x_j)} .
	\]
%\begin{align*}
%	\parens{a^{n-2}_{\geq t-2}} \sum_{x_i, x_j \in A}  {\alpha(x_i)\otimes\alpha^*(x_j)}, \qquad \parens{a^{n-2}_{\geq t-2}} \sum_{x_i, x_j \in A}  {\alpha(x_i)\otimes\alpha^*(x_j)},\\ \parens{a^{n-2}_{\geq t-2}} \sum_{x_i, x_j \in A}  {\alpha(x_i)\otimes\alpha^*(x_j)},\qquad \parens{a^{n-2}_{\geq t-2}} \sum_{x_i, x_j \in A}  \parens{\alpha(x_i)\otimes\alpha^*(x_j)}
%\end{align*}

	 We compute each of these now. Recall that $\sum_{x \in A} \alpha(x) = |G|\cdot \widehat{\ones_A}(\alpha)= -\sum_{x \in A^c} \alpha(x)$. Using this readily gives us,
%	 	\begin{align*}
%	|G^n|\cdot \widehat{\thresh} (\vec \rho) ~&=~ \sum_{r=0}^{2}a^{n-2}_{\geq t-2}\sum_{|\{x_i, x_j \}\cap A| = r}  \parens{\alpha(x_i)\otimes\alpha^*(x_j)} 
%	\end{align*}
%		\begin{align*}
%	\sum_{x_i, x_j \in A}  \alpha(x_i)\otimes\alpha^*(x_j)  \sum_{\vec{x_{ij}} \in G^{n-2}}   \thresh{\vec x} ~&=~   \alpha(A)\otimes\alpha^*(A) \sum_{\vec{x_{ij}} \in G^{n-2}}   \threshold{t-2} (\vec x_{ij}) \\
%	~&=~   \alpha(A)\otimes\alpha^*(A) \cdot a^{n-2}_{\geq t-2} \\
%	~&=~  |G|^2\cdot \widehat{\ones_A}(\alpha) \otimes \widehat{\ones_A}(\alpha^*)  \cdot a^{n-2}_{\geq t-2}  .
%	\end{align*}
%
%Additionally, $\alpha(A^c) = - \alpha(A) $. Using this and doing the same computation as above, we get that,  
\begin{align*}
|G^n|\cdot \widehat{\thresh} (\vec \rho)  	~=~  |G|^2\cdot \widehat{\ones_A}(\alpha) \otimes \widehat{\ones_A}(\alpha^*) \cdot \brackets{a^{n-2}_{\geq t-2} - 2\cdot a^{n-2}_{\geq t-1} + a^{n-2}_{\geq t} }  .
	\end{align*}
To finish the proof, we do the following computation, 
\begin{align*}a^{n-2}_{\geq t-2} - 2\cdot a^{n-2}_{\geq t-1} + a^{n-2}_{\geq t} ~&=~ \parens[\big]{a^{n-2}_{\geq t-2} -  a^{n-2}_{\geq t-1}} - \parens[\big]{a^{n-2}_{\geq t-1} - a^{n-2}_{\geq t}} \\
~&=~ \parens[\big]{a^{n-2}_{t-2}} - \parens[\big]{a^{n-2}_{t-1}} \qedhere 
\end{align*}	
\end{proof}

\begin{proposition}\label{prop:constant_comp}
If $|A| = \frac{|G|}{2}$ and $t = \frac{n+1-\sqrt{n}}{2}$, then
\[
C_{A,n,t} :=  \parens[\bigg]{\frac{{a_{t-1}^{n-2} - a_{t-2}^{n-2}}}{|G|^{n-2}}} ~\geq~ \Omega\parens[\Big]{\frac{1}{n-1}}.
\]	
\end{proposition}
\begin{proof} The proof is a simple calculation that we carry out below.
\begin{align*}
 C_{A,n,t} ~&=~ \parens[\bigg]{\frac{{a_{t-1}^{n-2} - a_{t-2}^{n-2}}}{|G|^{n-2}}} 	\\
  ~&=~ \parens[\bigg]{\frac{{{n-2 \choose t-1} \abs{A}^{t-1}\abs{A^c}^{n-1-t} - {n-2 \choose t-2} \abs{A}^{t-2}\abs{A^c}^{n-t}}}{|G|^{n-2}}} \\
%    ~&=~ \frac{1}{2^{n-2}} \parens[\bigg]{ {n-2 \choose t}  - {n-2 \choose t-1} } \\
    ~&=~ \frac{1}{2^{n-2}} \parens[\bigg]{ {n-2 \choose t-1}  - {n-2 \choose t-2} } \\
    ~&=~ \frac{1}{2^{n-2}} \frac{1}{n-1} \cdot {n-1 \choose t-1}  \cdot {{{\frac{(n+1-2t)}{2} }}} \\
    ~&\geq~ \frac{1}{2^{n-2}} \frac{1}{n-1} \cdot \frac{\Omega(2^{n})}{\sqrt{n}}  \cdot {\sqrt{n}} ~\geq~ \Omega\parens[\Big]{\frac{1}{n-1}}.  
    \end{align*}
The last inequality follows from our choice of $t = \frac{n+1-\sqrt{n}}{2}$.  for which, ${n-1 \choose t-1} = \Omega(\frac{2^n}{\sqrt{n}})$. 
\end{proof}

\begin{claim}[Lower Bound on fooling level-2 component] \label{claim:level2_lower}Let $X$ be a pseudo Cayley graph such that all non-trivial eigenvalues are equal to $\lambda <1/2$. Let $A$ be any subset of $G$. Then for the level-2 component of the threshold function, the following holds:
\[ \abs[\big]{\cE_X\parens[\big]{(\thresh)_2}} ~\geq~ (n-2)\cdot \lambda \cdot  \sfC_{A,n,t} \cdot \mu_A\cdot \mu_{A^c} \;\mper\] 
\end{claim}

%\begin{claim} If $A$ is a union of conjugacy classes, and $X$ is a quasi-Abelian Cayley graph, then the following holds:
%\[ \abs[\big]{\cE_X\parens[\big]{\ones^A_{>t}(2)}} ~\geq~ (n-1)\cdot \lambda_\minn(X) \cdot  \sfC_{A,n,t} \cdot \mu_A\cdot \mu_{A^c} \] 
%\end{claim}

\begin{proof}

We will use the notation $\alpha_{ij}$ to mean the irrep, $\alpha_{ij} = \rho_1\otimes \cdots\otimes\rho_n$ that is $\rho_i = \alpha$ and $\rho_j = \alpha^*$, and the rest are trivial. Since $X$ is a pseudo Cayley graph, we can apply, \cref{claim:level2_tensor}, Thus, for any $\vrho$ such that $|\vrho| = 2$ (\ie a level-2 irrep), we have $\cE_X(\vrho) = \sfM$ if $\vrho = \alpha_{ij}$ for any irrep $\alpha\in \Irrep(G)$. Else, it is $0$. Therefore, 
\begin{align*}
\cE_X\parens[\big]{(\thresh)_2}~&=~ \sum_{\vrho\in \mathrm{Irrep}(G^n),  |\vrho|= 2} d_\vrho\cdot \ip{\widehat{\thresh}(\vrho)}{ \cE_X(\vrho)}_{\HS} \\
~&=~\sum_{1\leq i<j \leq n} \sum_{\triv \neq \alpha\in \hat{G}} 
d_\vrho\cdot \ip{\widehat{\thresh}(\alpha_{ij})}{  \cE_X(\alpha_{ij})}_{\HS} && [X\text{ is quasi-Abelian}]
\\
~&=~\sum_{1\leq i<j \leq n} 
	 \sfC_{A,n,t}\cdot \lambda^{j-i}\cdot \sum_{\triv \neq \alpha\in \hat{G}}  d_{\alpha_{ij}}\ip{ 	\widehat{\ones_A}(\alpha)\otimes\widehat{\ones_A}(\alpha^*)}{ \sfM}_{\HS} && [\text{Using \cref{claim:level2fourier}}]
\end{align*}

For a fixed $i,j$, we can interpret the inner coefficent as the Fourier coefficient of $\ones_{A\times A} := \indicator{x_i, x_j \in A} : G\times G\to \C$. Thus, using ~\cref{corr:level2_proj} and \cref{eq:proj_ones}	, we get:
\begin{align*}
\sum_{\triv \neq \alpha\in \hat{G}}  d_{\alpha_{ij}}\ip{ 	\widehat{\ones_A}(\alpha)\otimes\widehat{\ones_A}(\alpha^*)}{ \sfM}_{\HS} ~&=~ P_G(\ones_{A\times A})(1,1) - \mu(\ones_{A\times A})\\
~&=~ \Ex{g}{(g,g) \in A\times A} - \mu(A)^2\\
~&=~ \mu(A)- \mu(A)^2 = \mu_{A}\cdot\mu_{A^c} .
\end{align*}

Plugging this back in out expression we get,
\begin{align*}
\cE_X\parens[\big]{(\thresh)_2}
~&=~\sum_{1\leq i<j \leq n} 
	 \sfC_{A,n,t}\cdot \lambda^{j-i}\cdot \mu_{A}\cdot\mu_{A^c} \\
~&=~	 \sfC_{A,n,t}\cdot  \mu_{A}\cdot\mu_{A^c} \cdot \lambda \cdot\parens[\Big]{n- \frac{1}{1-\lambda}}\cdot \frac{1-\lambda^n}{1-\lambda} \\
\abs[\big]{\cE_X\parens[\big]{(\thresh)_2}}~&\geq~
	 \sfC_{A,n,t} \cdot \mu_{A}\cdot\mu_{A^c}\cdot (n-2)\cdot  \abs{\lambda} \qquad \text{ if } \lambda < 1/2. 
\end{align*}
\end{proof}

%\tnote{Need to phrase it properly}

\begin{theorem}
Let $G$ be any finite group, $A \subseteq G$ such that $\frac{|A|}{|G|} = \frac{1}{2}$, and $X = \Cay(G^r, G^r\setminus\{1\})$ be the complete graph on $G^r$ without self-loops for some $r\geq 4$. Then for every $n$ large enough,
\[
\abs[\Big]{\cE_X\parens[\Big]{\threshold{\frac{n+1-\sqrt{n}}{2}}}} ~\geq~  \Omega \parens[\big]{\lambda(X)}.
\] 	
\end{theorem}
\begin{proof}
Using \cref{claim:fooling} we can separate the calculation into fooling the level-$2$ function and those beyond it, and thus for any function we have:  
\begin{align*}
	\cE_X\parens[]{f}  ~&=~ \sum_{i=2}^n \cE_X(f), \\
\abs[\big]{\cE_X(f)}  ~&\geq~ \abs[\big]{ \cE_X(f_2) } - \abs[\Big]{\sum_{k=3}^n \cE_X(f_k) } \;.
\end{align*}

We now let $f$ be the threshold function, $f = \threshold{\frac{n+1-\sqrt{n}}{2}}$. The graph $X$ has all non-trivial eigenvalues to be equal to $\frac{-1}{|G|^r} < 1/2$. We can then apply \cref{claim:level2_lower}, which when combined with \cref{prop:constant_comp}, we get $\abs[\big]{ \cE_X(f_2) } \geq \Omega(\lambda) $. To bound the remaining part, we use our upper bound~\cref{theo:main_fool} and obtain that, 
\[
\abs[\Big]{\sum_{i=3}^n \cE_X(f_i) } ~\leq~ 2\parens{16e|G|\lambda}^{\frac{3}{2}} \cdot \norm{f}_2~\leq~  o\parens[\Big]{\lambda^{\frac{3 (r-1)}{2 r}}} = o(\lambda) .\qedhere
\]
\end{proof}

%The only non-trivial claim to make is to observe that ~\cref{claim:level2_tensor} continues to hold on for quasi-Abelian Cayley graphs on $G^r$. This is because  
%
%One example of such a graph is a complete graph (without self-loops) on the group $G^3$ and label its vertices with the first coordinate. Note that this labeling is unbiased, and the graph is a quasi-Abelian Cayley graph as $S = G^3 \setminus \{1\}$, which is a union of conjugacy classes. 

\bibliographystyle{alphaurl}
\bibliography{macros,references}

\end{document}